\documentclass[lettersize,journal]{IEEEtran}

\usepackage{amsmath,amsfonts}
\usepackage{algorithmic}
\usepackage{algorithm}
\usepackage{array}
\usepackage{graphicx}

\usepackage{amsmath,amsfonts}
\usepackage{algorithmic}
\usepackage{array}
\usepackage{comment}
\usepackage{tcolorbox}

\usepackage{multirow}
\usepackage{amsthm}
\newtheorem{theorem}{Theorem}
\newtheorem{proposition}[theorem]{Proposition}
\newtheorem{lemma}[theorem]{Lemma}

\theoremstyle{definition}

\theoremstyle{remark}
\newtheorem{remark}[theorem]{Remark}

\usepackage{graphicx}
\usepackage{makecell}
\usepackage{booktabs}
\usepackage{algorithm}
\usepackage{algorithmic}
\usepackage{url}
\usepackage{textcomp}
\usepackage{stfloats}
\usepackage{url}
\usepackage{verbatim}
\usepackage{graphicx}
\usepackage{amsmath}
\usepackage{amsfonts}
\usepackage{amssymb}
\usepackage{amsthm}

\usepackage[shortlabels]{enumitem}

\def\BibTeX{{\rm B\kern-.05em{\sc i\kern-.025em b}\kern-.08em
    T\kern-.1667em\lower.7ex\hbox{E}\kern-.125emX}}
\begin{document}

\title{Staying Fresh: Efficient Algorithms for Timely Social Information Distribution}

\author{Songhua LI, Lingjie Duan,~\IEEEmembership{Senior Member,~IEEE}
\IEEEcompsocitemizethanks{\IEEEcompsocthanksitem Songhua Li and Lingjie Duan are with the Pillar of Engineering Systems and Design, Singapore University of Technology and Design, Singapore 487372.
 Email: \{songhua\_li,lingjie\_duan\}@sutd.edu.sg.}
}

\markboth{Journal of \LaTeX\ Class Files,~Vol.~14, No.~8, August~2021}%
{Shell \MakeLowercase{\textit{et al.}}: A Sample Article Using IEEEtran.cls for IEEE Journals}


\maketitle

\begin{abstract}
In location-based social networks (LBSNs), users sense urban point-of-interest (PoI) information in the vicinity and share such information with friends in online social networks. Given users' limited social connections
and severe lags in disseminating
fresh PoI to all, major LBSNs aim to enhance users' social PoI sharing by selecting $k$ out of $m$ users as hotspots and broadcasting their fresh PoI information to the entire user community. This motivates us to study a new combinatorial optimization problem that involves the interplay between an urban sensing network and an online social
network. We prove that this problem is NP-hard and also renders existing approximation solutions not viable. 
Through analyzing the interplay effects between the two networks, we successfully transform the involved PoI-sharing process across two networks to matrix computations for deriving a closed-form objective to hold desirable properties (e.g., submodularity and monotonicity). This finding enables us to develop a polynomial-time algorithm that guarantees a ($1-\frac{m-2}{m}(\frac{k-1}{k})^k$) approximation of the optimum. Furthermore, we allow each selected user to move around and sense more PoI information to share and propose an augmentation-adaptive algorithm with decent performance guarantees. 
Finally, our theoretical results are corroborated by our simulation findings using both synthetic and real-world datasets.
\end{abstract}


\maketitle

\section{Introduction}\label{section_introduction}
\IEEEPARstart{R}{ecent} years have witnessed the proliferation of sensor-rich mobile devices, which paves the way for people to collect and report useful point-of-interest (PoI) information that enhances their daily lives \cite{capponi2019survey,hamrouni2020spatial}.  For example, users of mobile devices, such as smartphones, can upload location-tagged photos and videos to online content platforms like Flickr \cite{flickrapp} and TikTok \cite{tictokapp}, allowing them to share personal experiences and help others discover new places, including popular restaurants, hidden landmarks, and local events and promotions.
They can report restaurant waiting times to online queue platforms, such as Yelp Waitlist \cite{yelpwaitlistsapp}, enabling others to plan their visits more efficiently and reducing overall waiting times. Users can check in and add new locations (e.g.,  businesses, restaurants, landmarks, etc.) to the map on life-logging platforms such as Gowalla \cite{Gowalladata} and Foursquare Swarm \cite{foursquareswarm}, enriching the shared map of PoI and supporting social interaction around new or less-known places. Additionally, users can also report traffic conditions (e.g., accidents, road closures, etc.) to online navigation platforms like Waze \cite{wazeapp} and Google Maps \cite{googlemaps},  aiding in traffic updates and helping others navigate more efficiently. 

The above online platforms, also known as location-based social networks (LBSNs), bridge urban sensing and online social networks for facilitating information collection and sharing among many users \cite{zheng2011location,bao2015recommendations}. As LBSN users are both PoI contributors and consumers, they share their PoI collections with friends for mutual information benefits \cite{zheng2011location,li2017dynamic}.
It is also worth noting that PoI sharing in LBSNs can be conducted in a privacy-preserving manner, enabled by a range of schemes developed in recent years. These schemes are often supported by federated learning-based approaches \cite{perifanis2023fedpoirec,guo2021prefer}, differential privacy techniques \cite{chen2020practical,huo2021privacy,chen2022differential,zhao2020local}, cryptographic approaches \cite{sun2024towards, zhang2025survey}, and other techniques \cite{hong2023location,hong2022protecting} for location privacy protection, which collectively ensure the security of PoI sharing and recommendations in LBSNs. Nonetheless, the majority of LBSN users typically have limited social connections, which hampers the overall efficiency of PoI sharing across the network. For example, a recent study \cite{chakravarthy2023nnpec} indicates that, among a total of 196,591 individuals, the average number of social connections for a Gowalla user is only under 24. Furthermore, it has been reported that sharing one's PoIs with friends of friends is time-consuming and experiences severe lags in covering the entire social network \cite{cha2009measurement}.  

Consequently, users' limited social connections and slow sharing speed hinder their exchange of a wide range of time-sensitive PoI information. To address this issue, major LBSNs often select certain users to broadcast their PoI information to the entire community, thereby enhancing social PoI sharing among all users. 
For instance, TikTok selectively showcases some location-tagged restaurant information posted by certain users in popular locations to others, catering to the preferences of the majority \cite{tiktokrestaurant}.

Besides the communication overhead above, information overload can exhaust users and even disrupt their continued engagement in LBSN platforms \cite{fu2020social}. Notably, LBSN users are often influenced by the initial pieces of PoI information broadcast by the system \cite{fu2020social}, a phenomenon known as the anchoring effect \cite{seo2021point}. Consequently, LBSNs typically broadcast only a small amount of PoI information to all users \cite{tiktokrestaurant}. 
Therefore, we are motivated to ask the following question.

\textbf{Question 1:} \textit{When an urban sensing network meets an online social network, how can an LBSN
optimally select $k$ out of a total of $m$ users as hotspots to track and broadcast their PoI information to the user community? 
}

We aim to reveal the ramifications of Question~1 from a theoretical perspective and pursue performance-guaranteed solutions. 
When finding $k$ out of $m$ users to track and broadcast
their locally sensed PoI to all, one may expect the combinatorial nature of the problem. Thereby, we first survey the most related classical problems, which include the social information maximization (SIM) problem in social networks and the covering problems in sensing networks. 

The SIM problem seeks a subset $S\subseteq V$ of $k$ user nodes in a given social graph $G=(V, E)$
to maximize the influence spread in graph $G$, which is known to be NP-hard and can be approximated to a ratio $(1-\frac{1}{e})$ of the optimum \cite{kempe2003maximizing,li2022influence}, where $e$ is the base of natural logarithms. Note that the SIM problem solely focuses on optimization in a single social network, or simply assumes that each social user collects an equal amount of non-overlapping PoI in the urban-sensing network. Differently, users' PoI collections in our problem may exhibit content overlap and volume variation. Hence, it becomes critical to take into account the interplay effect between the social and urban-sensing graphs for optimization purposes in our problem. 

Related covering problems mainly include the vertex cover and set cover problems, both of which are NP-hard \cite{hochbaum1982approximation,karakostas2009better}. Given a graph $G(V,E)$, the vertex cover problem (VCP) aims to find the least vertices to include at least one endpoint of each edge in $E$. State-of-the-art approximation ratio for VCP achievable in polynomial time is $2-\Theta(\frac{1}{\sqrt{\log |V|}})$ \cite{karakostas2009better}.
Given a ground set $\mathcal{U}=\{1,2,...,n\}$ of elements, a collection $\mathcal{S}=\{S_1,...,S_m\}$ of $m$ sets such that $\bigcup\limits_{S_i\in \mathcal{S}}S_i=\mathcal{U}$, and a budget $k(<m)$, the set cover problem (SCP) aims to find a sub-collection of $k$ sets from $\mathcal{S}$ to cover as many elements in $\mathcal{U}$ as possible. State-of-the-art approximation guarantee for SCP is known to be $(1-\frac{1}{e})$ \cite{feige1998threshold}. 
However, both VCP and SCP problems largely overlook the effects of social information sharing. Applying their solutions to our optimization problem, which integrates the sensing and social networks, results in significant efficiency loss as demonstrated later in Section~\ref{section_experiment}.

In addition, another line of related works focuses on optimizing information dissemination within social networks, where information spreads progressively from selected users through their social ties, in contrast to the simultaneous dissemination to all users considered in this paper. For instance, Lu et al. \cite{lu2015towards} propose a probabilistic model to characterize information dissemination in social networks, accompanied by a set-cover-based algorithm. Later, Hsu et al. \cite{hsu2019scheduling} further study fresh information dissemination on wireless broadcast networks, where a base station updates some network users on timely information generated randomly from certain given sources. Based on Markov Decision Process (MDP) techniques, Hsu et al. \cite{hsu2019scheduling} propose a scheduling algorithm to maintain information freshness among users in the network. While these works demonstrate favorable empirical performances through their extensive simulations, their algorithms do not come with any performance guarantee. Li et al. \cite{li2024age} examine a multi-stage information dissemination process under a deterministic diffusion model and develop a performance-guaranteed algorithm for maximizing the freshness of information received by each user in the network. For more discussions about this line of related works, we refer readers to the survey work \cite{banerjee2020survey}. Nonetheless, these works overlook the relative positioning of LBSN users, i.e., their connectivity within the urban sensing network, which determines the degree of PoI overlap among users. This overlap plays a crucial role in the user selection process for effective social information sharing. Thereby, ignoring urban sensing dynamics greatly reduces the effectiveness of their solutions for the problem to be addressed in this paper

The key challenge to answer Question~1 is outlined below.

\textbf{Challenge 1.} \textit{The urban sensing network is inherently different from the online social network, making it difficult to formulate a cohesive objective to maximize the PoI information accessibility to users. Indeed, this problem is NP-hard as proved in Section~\ref{sec_problemstatement}, and renders existing solutions unviable.}

%
Further, we extend our study to a more general scenario where each selected user is granted the capability to move around to sense more PoI information. Suppose each selected user could move $n$ hops forward from her current location and gather PoI information along her way, another question naturally arises as follows:

 \begin{table*}[t]
  \caption{Main result summary, where $\varpi$ and $m$   denote node sizes in the urban sensing and social graphs, respectively, and {\color{black}$\tau$ denotes the number of $n$-hop paths in $G_1$ that start from a user node}.}
  \begin{center}
 \begin{tabular}{|c|c|c|c|}
    \hline
    \textbf{Scenarios}& \textbf{Hardness}& \textbf{Approximation Guarantee} & \textbf{Time Complexity}\\
    \hline
\makecell[c]{Static crowd-sensing \\(see \textbf{Section} \ref{sec_fundamentalalg})}&\makecell[c]{NP-hard\\(see \textbf{Proposition}~\ref{nphardnessingeneral})}& \makecell[c]{$1-\frac{m-2}{m}(\frac{k-1}{k})^k$\\(see \textbf{Theorem} \ref{ub_OBJ1F})} &  \makecell[c]{$O(k(m-k)m^{3.372})$\\(see \textbf{Theorem} \ref{ub_OBJ1F})}\\
 \hline
\multirow{2}*{
\makecell[c]{Mobile crowd-sensing\\(see \textbf{Section} \ref{sec_resourceaugmented})}}&\multirow{2}*{\makecell[c]{NP-hard\\(see \textbf{Proposition}~\ref{nphardnessingeneral})}}&  \makecell[c]{$\frac{1}{k}-\frac{\varpi-2}{\varpi}\frac{(k-1)^k}{k^{k+1}}$ with no-augmentation,\\$\frac{g}{k}[1-\frac{\varpi-2}{\varpi}(\frac{k-1}{k})^k]$ with $g$-augmentation,\\(see \textbf{Theorem} \ref{theorem_n_hop_forward})} &
\makecell[c]{{\color{black}$O(k(\tau-k)\varpi^{2.372})$}\\(see \textbf{Theorem} \ref{theorem_n_hop_forward})}
\\
\cline{3-4}
~&~&\makecell[c]{$1-\frac{\varpi-2}{\varpi}(\frac{k-1}{k})^k${when each sensing node}\\ is associated with a user. 
(\textbf{Theorem} \ref{improvedratio_nhop})}&\makecell[c]{{\color{black}$O(k(\tau-k)\varpi^{2.372})$}\\ {(see \textbf{Theorem} \ref{improvedratio_nhop})}}\\
\hline
    \end{tabular}
\end{center}
    \vskip -0.1in
    \label{mainresult_tab}
\end{table*}

 \textbf{Question~2:} \textit{Given the initial locations of users in the urban sensing network and their social connections, how can an LBSN jointly select $k$ out of $m$ users and design their $n$-hop paths for better collecting and sharing PoI information?}
 
The literature related to Question~2 primarily focuses on designing a single route for a single user, e.g., the well-known traveling salesman problem (TSP) \cite{christofides1976worst} and the watchman problem \cite{livne2022optimally}. However, these problems ignore the social sharing effects and focus solely on distance minimization objectives, which are distinct from ours. For example, the TSP requires visiting every node exactly once and returning to the origin, while the watchman problem seeks the shortest path/tour in a polygon so that every point is visible from some point on the path \cite{mitchell2013approximating}.
Technically, the interplay effects between the sensing and social networks in our joint optimization expand the search space of our problem, rendering existing solutions unfeasible. Other related works mainly focus on mobile crowd-sourcing  \cite{ganti2011mobile,wang2018multi,lai2022optimized, xu2022approximation,zhou2022online, wang2022dynamic}, PoI recommendations \cite{bao2015recommendations,seo2021point,wang2018exploiting}, and energy-efficient optimization in mobile networks
\cite{sun2021time,sun2019energy}. However, these works also overlook the critical aspect of information sharing within online social networks. Moreover, recent works on users' data sharing in social networks \cite{li2017dynamic,gong2017social} either fail to take the urban sensing network into consideration or lack approximation guarantees for their solutions.

The key challenge posed by Question~2 is outlined below.

 \textbf{Challenge 2.} \textit{
Both the intricate topology of a general sensing network and the involved PoI sharing process in the social network make it NP-hard to
coordinate $m$ individual $n$-hop paths for those selected users to achieve the maximum PoI accessibility among users. }

%

%


Our key novelties and main contributions are summarized below and also in Table~\ref{mainresult_tab}. Due to page limits, omitted proofs can be found in our appendix as supplementary material.  
\begin{itemize}
    \item 
 \textit{\color{black}New Combinatorial Problem to Enhance Social Sharing of Fresh PoI Information.}  When urban sensing meets social sharing, we introduce a new combinatorial optimization problem to select $k$ out of $m$ users to maximize the accessibility of PoI information to all users. We develop, to the best of our knowledge, the first theoretical foundations in terms of both NP-hardness and approximation guarantee for the problem. We practically allow users to have different PoI preferences and move around to sense more PoI to share.
\item
\textit{Polynomial-time Approximation Algorithm with Performance Guarantees}.
Through meticulous characterization 
of the interplay effects between social and sensing networks, we successfully transform
the involved social PoI-sharing process to matrix multiplication/computation, which serves as a building block for formulating our tractable optimization objective. Additionally, this transformation also allows for various adaptations in our subsequent solutions. In a fundamental scenario where users are static and only collect PoI information within their vicinity, we present a polynomial-time algorithm that successfully guarantees an approximation $1-\frac{m-2}{m}(\frac{k-1}{k})^k>0.632$ of the optimum. This success is attributed to the desirable properties identified in our objective function.
\item
\textit{Resource-augmented Algorithm for Mobile Crowd-sensing with $n$-hop-forward.} We extend our approximate framework to address a more general scenario where each selected user could further move along a $n$-hop-path to mine and share more PoI information. To this end, we circumvent the routing intricacies by transforming the problem into an optimization alternative, fitting within our approximate framework for a simplified scenario of static crowd-sensing. Further, we introduce a novel resource-augmentation technique and propose an augmentation-adaptive algorithm that guarantees bounded approximations. These approximation guarantees range from $\frac{1}{k}(1-\frac{1}{e})$ to $1-\frac{1}{e}> 0.632$ when adjusting the augmentation factors from one to the budget $k$. The effectiveness of our algorithms is demonstrated by simulations using both synthetic and real-world datasets across different network topologies.
\end{itemize}

The rest of this paper is organized as follows. Section~\ref{sec_problemstatement} presents the system model for social PoI sharing, integrating the sensing and social networks, along with the NP-hardness proof for the problem. Section~\ref{sec_fundamentalalg} provides a polynomial-time approximation algorithm for the fundamental scenario of static crowd-sensing, and 
Section~\ref{sec_resourceaugmented} extends the approximate scheme to a generalized scenario of mobile crowd-sensing. Section~\ref{section_experiment} corroborates our theoretical findings with simulations. Finally, Section~\ref{section_conclusion} concludes this paper.
\section{Problem Statement And Preliminaries}\label{sec_problemstatement}
\begin{figure}[!t]
    \centering
\includegraphics[width=8cm]{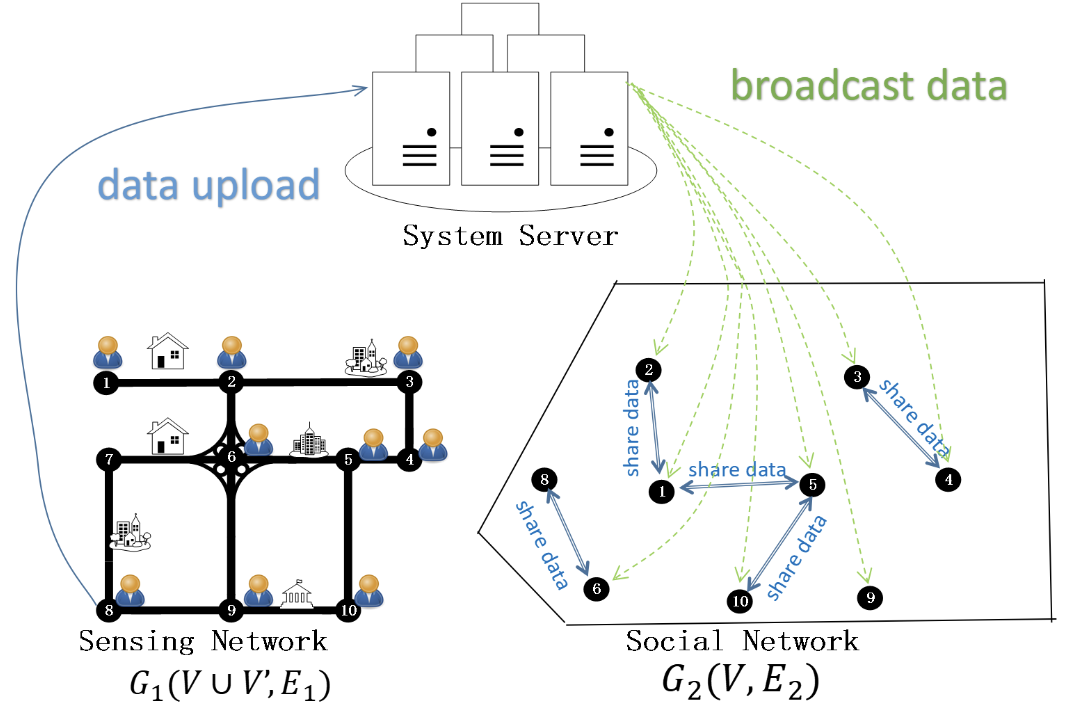}
    \caption{An illustration example of choosing $k=1$ hotspot to show the schematic setting of an LBSN system, where $m=9$ users (at nodes 1-6, 8-10) sense the fresh PoI information around from the incident edges in the sensing graph $G_1$. Then, they share such PoIs with their friends in the online social graph $G_2$. To enhance social PoI sharing, the system selects user node 8 as a hotpot to track and broadcast her PoI collection to all users immediately.}
    \label{systemfigure}
\end{figure}
We consider a budget-aware location-based social network (LBSN) that involves $m$ users as in Fig.~\ref{systemfigure},
where each user collects PoI information from her vicinity in an urban sensing graph $G_1$ and shares her PoI collection with her immediate friends/neighbors in another online social graph $G_2$. 
Suppose, \textit{w.l.o.g.}, that the sensing graph $G_1=(V\cup V',E_1)$ consists of two sets of nodes representing distinct locations, in which a location node in $V=\{v_1,...,v_m\}$ has a user while a location node in $V'$ does not. We also refer to a node $v_i\in V$ that has a user as a \textit{user node}. Totally, there are $\varpi$ location nodes in $G_1$, i.e., $|V\cup V'|=\varpi$. An edge in $E_1$ of $G_1$ indicates a specific road (with PoIs to be sensed) connecting the two location nodes. Different edges/roads have distinct PoI content but share equal PoI weight. The cumulative PoI information gathered from any edge subset of $E_1$ is determined by the count of distinct edges within that subset. The social graph is defined over the same group of $m$ users and is denoted as $G_2(V,E_2)$, where two user nodes will be connected by an edge in $E_2$ if and only if they are immediate friends.

For each graph $G_i\in\{G_1,G_2\}$ and a node $v$ in $G_i$, we use $N_i(v)$ and $E_i(v)$ to denote the set of $v$'s neighbors in $G_i$ and the set of edges in $G_i$ that are incident to the node $v$, respectively. For any subset $\overline{V}\subseteq V\cup V'$ of nodes, 
$N_1(\overline{V})$ 
is the set of nodes that are directly connected to at least one node in $\overline{V}$, i.e., $N_1(\overline{V})\triangleq \bigcup_{v\in \overline{V}}N_1(v)$, and $E_1(\overline{V})$ is the set of edges that are incident to at least one node in $\overline{V}$, i.e., $E_1(\overline{V})\triangleq \bigcup_{v\in \overline{V}}E_1(v)$. 

Since users prefer fresh PoI information without severe propagation lag, we assume for simplicity that users share their PoI collections only with their immediate neighbors in the social graph $G_2$. Nonetheless, our solutions in Section~\ref{subsection_computing_phi} can be easily adapted to accommodate scenarios where PoI sharing extends over multiple hops in the social graph. 
To enhance the social PoI sharing, the system
will select $k$ out of $m$ users to aggregate and broadcast their PoI collections to all users. Fig.~\ref{systemfigure} illustrates a case where $k=1$, with the PoIs of user node 8 being broadcast to all users. Generally, we denote $s_j\in V$ as the $j$th user selected by the LBSN system server and summarize the first $j$ selected users in set $S_j\triangleq \{s_1,...,s_j\}$. 

Thereby, the fresh PoI information available to a user $v_i\in V$ is comprised of the following three parts:
\begin{itemize}
    \item Part~1 is her own PoI collection;
    \item Part~2 contains those PoIs collected and shared by $v_i$'s neighbors in the social graph $G_2$;
    \item  Part~3 consists of those PoIs that are collected by the $k$ selected users and thus broadcast by the LBSN system. 
\end{itemize}
As outlined below, we examine two typical crowd-sensing scenarios: static crowd-sensing and mobile crowd-sensing. We formulate the objectives for each scenario separately.
\subsection{Two Typical Scenarios for PoI Sensing And Sharing}\label{settings_description}
\subsubsection{Static crowd-sensing scenario}
To answer \textbf{Question~1} introduced earlier in Section~\ref{section_introduction}, we study the \textit{static crowd-sensing scenario} that considers a fundamental case where users are static and only collect PoI information from edges incident to their static locations. When every user $v_i\in V$ is interested in the PoI information from the whole sensing graph $G_1$ for future utilization, we further refer to the scenario as the \textit{static crowd-sensing without PoI preference}. With her own PoI collection $E_1({v_i})$, those PoIs collected and shared by her neighbors (i.e., her immediate friends) in $N_2(v_i)$, and those PoIs collected/sensed by the $k$ selected users in $S_k$ and broadcast by the LBSN, the utility function \footnote{We can easily extend (\ref{obj_function_pre}) to include different weights for different PoI parts, which does not change our later algorithm design and bound analysis.} for a user $v_i\in V$ without PoI preference is defined as follows:
\begin{equation}\label{obj_function_pre}
    \phi_i(S_k)\triangleq|E_1(\{v_i\}\cup N_{2}(v_i)\cup S_k)|,
\end{equation}
in which $E_1(\{v_i\})$, $E_1(N_{2}(v_i))$, and $E_1(S_k)$ refer to the aforementioned parts 1, 2, and 3 of PoI information, respectively.
When a user $v_i\in V$ is just interested in the PoI information along edges in a certain subset $E_1^{(i)}\subseteq E_1$ for her recent use, we further refer to the scenario as the \textit{static crowd-sensing with PoI preference}.  Then,
the PoI utility function for a user $v_i\in V$ with PoI preference can be updated from (\ref{obj_function_pre}) to 
\begin{equation}\label{obj_function_pre_restricted}
\phi_i(S_k)\triangleq|E_1^{(i)}\cap E_1(\{v_i\}\cup N_{2}(v_i)\cup S_k)|.
\end{equation}
%
\subsubsection{Mobile crowd-sensing scenario} To answer \textbf{Question~2} introduced in Section~\ref{section_introduction}, we further consider the \textit{mobile crowd-sensing scenario}, which extends to a more general case where each selected user is allowed to move along a $n$-hop-path from her current location to sense more PoI information to share \cite{capponi2019survey,suhag2023comprehensive}.
We denote the path that the LBSN system recommends to a 
selected user $s_j$ as $p_j=(s_j,s_{j1},s_{j2},...,s_{jn})$,
and summarize those $n$-hop-paths for all the $k$ selected users in the set $P_k\triangleq \{p_1,...,p_k\}$. A selected user is willing to follow the recommended path simply because the LBSN system can provide her with a significant incentive (e.g., monetary rewards, virtual credits, reputation points, etc) \cite{suhag2023comprehensive}. Again, due to budget limits, the LBSN system selects $k$ mobile users in total for PoI collection and sharing. 
To simplify notation, we use $N_1(P_k)$ to denote the set of nodes in $G_1$ that are included in at least one path in $P_k$. 

Now, the PoI collection by a selected user $s_j\in S_k$ becomes $E_1(\{v_j\}\cup N_1(p_j))=E_1(N_1(p_j))$, which consists of those edges incident to the nodes that she visits along $p_j$. By aggregating those PoI collections of all the $k$ selected users, the following PoI aggregation will be broadcast to the public:
\begin{equation*}
    \bigcup_{j=1}^k E_1(N_1(p_j))=E_1(\bigcup_{j=1}^kN_1(p_j))=E_1(N_1(P_k)).
\end{equation*}
Given the $k$ selected users in $S_k$, 
the PoI utility function for a user $v_i\in V$ without PoI preference is updated from (\ref{obj_function_pre}) to
\begin{equation}\label{obj_function_pre_nhop_forwarded}
     \phi_i(S_k)\triangleq|E_1(\{v_i\}\cup N_{2}(v_i)\cup N_1(P_k))|.
\end{equation}
When accounting for PoI preferences, the utility function can be modified from (\ref{obj_function_pre_restricted}) by substituting the term $S_k$ with $N_1(P_k)$.

%

In both static and mobile crowd-sensing scenarios, the average PoI utility for a user, denoted by $\Phi(S_k)$, is viewed as \textit{welfare} and defined as
\begin{equation}\label{obj_function}  \Phi(S_k)\triangleq\frac{1}{m}\sum\limits_{v_i\in V}\phi_i(S_k).
\end{equation}
We notice that $|E_1|$ provides a natural upper bound for $\phi_i(S_k)$ and, consequently, $\Phi(S_k)$.

\textit{Our objective} is to find a subset $S_k\subseteq V$ of $k$ users that maximizes $\Phi(S_k)$, i.e.,  
\begin{equation}\label{initial_objective_globally}
\max\limits_{S_k\subseteq V} \Phi(S_k).
\end{equation} Particularly, in the mobile crowd-sensing scenario, our objective (\ref{initial_objective_globally}) also requires
jointly determining the optimal $n$-hop paths for the $k$ selected users. 

By a non-trivial reduction from the well-known vertex cover problem \cite{hochbaum1982approximation},
we show in the following proposition that our social-enhanced PoI sharing problem in (\ref{initial_objective_globally}) is NP-hard.

\begin{proposition}\label{nphardnessingeneral}
The social-enhanced PoI sharing problem (\ref{initial_objective_globally}) is NP-hard for both static and mobile crowd-sensing scenarios.
\end{proposition}
The NP-hardness of our problem indicates that no optimal solution can be achieved in polynomial time. Hence, we aim at efficient approximation algorithms, which are often evaluated by approximation ratios as introduced below.
\subsection{Approximation Ratio for Performance Evaluation}
  The approximation ratio is a standard metric to evaluate the performances of approximation algorithms, which refers to the worst-case ratio between an algorithm's welfare and the optimal welfare \cite{williamson2011design,xu2022approximation}. Formally, given an instance $I=(G_1, G_2)$ of the problem, let $S_k^*(I)$ denote the set of users selected by an optimal solution. The approximation ratio, denoted as $\rho$, is defined as the ratio of our solution's welfare (with $S_k(I)$ of our selected users) over the optimal welfare, which is taken over all possible instances. Formally, it is
\begin{equation}\label{def_approratio}
\rho\triangleq\inf\limits_{I=(G_1, G_2)}\frac{\Phi(S_k(I))}{\Phi(S_k^*(I))}.
\end{equation}
When there is no ambiguity in the context, we abuse notation to use $S_k$ and $S_k^*$ to denote $S_k(I)$ and $S_k^*(I)$, respectively.

In the following Section~\ref{sec_fundamentalalg}, we focus on algorithm design in the static crowd-sensing scenario as a warm-up, built upon which we will further advance our approximation algorithms for the mobile crowd-sensing scenario later in Section \ref{sec_resourceaugmented}.
\section{Approximation Algorithm for Static Crowd-sensing}\label{sec_fundamentalalg}
  To address the static crowd-sensing problem in (\ref{initial_objective_globally}), we introduce Algorithm \ref{greedyalgforfiniteidea1}, referred to as GUS. The key idea of GUS is to iteratively select each of the $k$ users from $V$ in a way that maximizes the marginal welfare contributed to the system at each step.
\begin{algorithm}
\caption{\textsc{Greedy-User-Selection} (GUS)}\label{greedyalgforfiniteidea1}
\textbf{Input}:  $G_1=(V,E_1)$, $G_2=(V,E_2)$ and $k$.\\
\textbf{Output}: $S_k=\{s_1, ..., s_k\}$.
\begin{algorithmic}[1] 
\FOR{$s_i= s_1, ..., s_k$}
    \STATE Select $s_i\leftarrow\arg\max\limits_{v_x\in V}\left\{{\Phi}(S_{i-1}\cup\{v_x\})-{\Phi}(S_{i-1})\right\}$.
\ENDFOR
\end{algorithmic}
\end{algorithm}

Despite its simple looking in step 2, running Algorithm~\ref{greedyalgforfiniteidea1} requires inputting the precise value of each $\Phi(S_i)$ directly from the objective function (\ref{obj_function}). The computation method for $\Phi(S_i)$ also determines the time complexity of Algorithm \ref{greedyalgforfiniteidea1}. 

Benefiting from our characterizations of the interaction effects between the sensing network $G_1$ and the social network $G_2$, we develop an efficient method for computing $\Phi(S_i)$ by
translating the involved
social PoI sharing process into matrix computation. Of independent interest, our approach to computing $\Phi(S_i)$ also enables us to readily adapt our algorithms to a range of extensions beyond the scope of this paper's problems. 
%
\subsection{Tractable Approach for Computing $\Phi(S_i)$}\label{subsection_computing_phi}
  We begin with the basic case of  \textit{static crowd-sensing without PoI preference}, and then extend our approach to the case with PoI preference. 
\subsubsection{Static Crowd-sensing without PoI Preference}
We first look at $\Phi(\varnothing)$ when no one has yet been selected to broadcast her PoI collection. This is critical as it establishes a foundation for sequentially computing each $\Phi(S_{i})$ for $i\in \{1,2,...,k\}$. 
 Our approach for computing $\Phi(\varnothing)$ consists of the following four steps.

 \textbf{Step 1.} Construct the \textit{sensing} matrix $\textbf{A}=(a_{i,j})_{1\leq i,j\leq \varpi}$ based on the given sensing graph $G_1$: the entry
$a_{i,j}$, located at the $i$-th row and the $j$-th column of the matrix $\textbf{A}$, is set to one if there is an edge in the sensing graph $G_1$ that connects the two location nodes $v_i$ and $v_j$, i.e.,  
$(v_i,v_j)\in E_1$. If $(v_i,v_j)\notin E_1$, the entry $a_{i,j}$ of matrix $\textbf{A}$ is set to zero. Since $G_1$ does not have a self-loop that connects a node to itself, entries in the main diagonal of matrix $A$ are all set as zero. Consequently,
\begin{equation}\label{metrixa_eq}
    a_{i,j}=\left\{\begin{matrix}
1, & {\rm if\;}i\neq j{\rm \;and\;} (v_i,v_j)\in E_1,\\ 
0,& {\rm otherwise}.
\end{matrix}\right.
\end{equation}
Note that by replacing the value of $a_{i,j}$ with the corresponding weight of edge $(v_i,v_j)\in V$, one can easily generalize our approach to the scenario where edge weights are not uniform. To illustrate this step of the algorithm, we consider the toy example in Fig.~\ref{systemfigure}, which yields the following sensing matrix:
\begin{equation}
\begin{small}
  \textbf{A}=  \begin{pmatrix}
  0&1  & 0 &0  &0  &0  &0  &0  &0  &0 \\
 1 &0  &1  &0  &0  &1  & 0 &0  &0  &0 \\
 0 &1  &0  &1  &0  & 0 &0  &0  &0  &0 \\
 0 & 0 &1  &0  & 1 & 0 &0  &0  &0  &0 \\
 0 &0  &0  &1  &0  &1  &0  &0  &0  &1 \\
 0 &1  &0  &0  &1  &0  &1  &0  &1  &0\\
 0 & 0 &0  &0  &0  &1  & 0 &1  &0  &0\\
  0& 0 &0  &0  &0  &0  &1  & 0 &1  &0  \\
 0 & 0 &0  &0  &0  &1  & 0 &1  &0  &1\\
  0&0  &0  &0  &1  &0  &0  & 0 &1  &0
\end{pmatrix}.
\end{small}
\end{equation}

 \textbf{Step 2.} Construct the social matrix  $\textbf{B}=(b_{i,j})_{1\leq i,j\leq \varpi}$ which reflects the social relationships among those $m$ users.  Due to technical reasons, matrix $\textbf{B}$ involves not only those $m$ user nodes in the set $V$ but also those $\varpi-m$ non-user nodes in the set $V'$, resulting in a $\varpi\times \varpi$ matrix. All nodes are indexed in the same way as in the sensing matrix \textbf{A}. However, in matrix $\textbf{B}$, only entries corresponding to user nodes in $V$ are set to one, while all entries corresponding to non-user nodes in $V'$ are set to zero. Since a user can naturally access her own PoI collection, entries on the main diagonal of matrix $\textbf{B}$ that correspond to user nodes are set to one. Specifically, the entry $b_{i,j}$ at the $i$-th row and the $j$-column of matrix $B$ follows 
\begin{equation}\label{matrixb_eq}
    b_{i,j}=\left\{\begin{matrix}
1,& \big(v_i,v_j\in V\big) {\rm \;and\;} \big((v_i,v_j)\in  E_2\;{\rm or\;} i=j\big),\\ 
0, & {\rm otherwise\;}.
\end{matrix}\right.
\end{equation}
Note that both $\textbf{A}$ and $\textbf{B}$ are symmetric matrices. Based on this, the social matrix corresponding to the toy example in Fig.~\ref{systemfigure} is constructed as follows:
\begin{equation}
\begin{small}
   \textbf{B}= \begin{pmatrix}
  1&1  & 0 &0  &1  &0  &0  &0  &0  &0 \\
 1 &1  &0  &0  &0  &0  & 0 &0  &0  &0 \\
 0 &0  &1  &1  &0  & 0 &0  &0  &0  &0 \\
 0 &0  &1  &1  &0  & 0 &0  &0  &0  &0 \\
 1 & 0 &0  &0  & 1 & 0 &0  &0  &0  &1 \\
 0 &0  &0  &0  &0  &1  &0  &1  &0  &0 \\
 0 &0  &0  &0  &0  &0  &1  &0  &0  &0\\
 0 & 0 &0  &0  &0  &1  & 0 &1  &0  &0\\
 0 & 0 &0  &0  &0  &0  & 0 &0  &1  &0\\
  0&0  &0  &0  &1  &0  &0  & 0 &0  &1
\end{pmatrix}.
\end{small}
\end{equation}

 \textbf{Step 3.} Compute the matrix product of $\textbf{A}$ and $\textbf{B}$, denoted as $\textbf{C}=(c_{i,j})_{1\leq i,j,\leq m}$, i.e., 
$\textbf{C}\triangleq\textbf{AB}$. For each user node $v_x\in V$, let $\sigma_x(B)$ denote the set summarizing all indices of those zero entries in the $x$-th column of the matrix $\textbf{B}$.  
Construct the minor of matrix $\textbf{A}$ by removing from $\textbf{A}$ the rows and columns indexed by $\sigma_x(\textbf{B})$. This minor is denoted as $\textbf{M}_{\sigma_x(\textbf{B})}$. Intuitively, $\textbf{M}_{\sigma_x(\textbf{B})}$ includes those PoI collections by either the user $v_x$  or $v_x$'s friends. We then compute the matrix product of $\textbf{A}$ and $\textbf{B}$ for the toy example in Fig.~\ref{systemfigure}, yielding:
\begin{equation}
\begin{small}
   \textbf{C}= \begin{pmatrix}
  1&1  & 0 &0  &0  &0  &0  &0  &0  &0 \\
 1 &1  &1  &1  &1  &1  & 0 &1  &0  &0 \\
 1 &1  &1  &1  &0  & 0 &0  &0  &0  &0 \\
 1 &0  &1  &1  &1  & 0 &0  &0  &0  &1 \\
 0 & 0 &1  &1  & 1 & 1 &0  &1  &0  &1 \\
 2 &1  &0  &0  &1  &0  &1  &0  &1  &1 \\
 0 &0  &0  &0  &0  &2  &0  &2  &0  &0\\
 0 & 0 &0  &0  &0  &0  & 1 &0  &1  &0\\
 0 & 0 &0  &0  &1  &2  & 0 &2  &0  &1\\
  1&0  &0  &0  &1  &0  &0  & 0 &1  &1
\end{pmatrix}.
\end{small}
\end{equation}
Next, we consider user $v_6$ (i.e., $x=6$) as an example to illustrate the computation of her initial PoI utility $\phi_6(\varnothing)$ based on our approach. Specifically, we have $\sigma_6(\textbf{B}) = \{1,2,3,4,5,7,9,10\}$. Further, the corresponding minor of matrix $\textbf{A}$ can be derived, by deleting its rows in $\sigma_6(\textbf{B})$ and columns in $\sigma_6(\textbf{B})$, as $\textbf{M}_{\delta_6(\textbf{B})}= \begin{pmatrix}
  0   &0  \\
  0    & 0
\end{pmatrix}$.

 \textbf{Step 4.} Compute 
$\sum\limits_{1\leq j\leq m}c_{j,x}-\frac{1}{2}\textbf{e}^T_x \textbf{M}_{\sigma_x(B)} \textbf{e}_x$, where $\textbf{e}_x$ is a column vector of ones with dimensions suitable for the matrix multiplication.  Finally, based on Step~4, the resulting value for user $v_6$ (i.e., $x = 6$) is determined as follows:
\begin{equation}\label{toy_initialutility}
    \sum\limits_{1\leq j\leq 10}c_{j,6}-\frac{1}{2}\textbf{e}^T_6 \textbf{M}_{\sigma_6(B)} \textbf{e}_6=6.
\end{equation}

Due to the PoI utility function $\phi_i(S_k)$ in (\ref{obj_function_pre}), we can obtain $\phi_i(\varnothing)=|E_1(\{v_x\}\cup N_1(v_x))|$ as in the following lemma.

\begin{lemma}\label{lemma_01}
 For each $v_x\in V$, the following holds: 
 \begin{equation}\label{lemma1_formulation}   \phi_i(\varnothing)=\sum\limits_{1\leq j\leq m}c_{j,x}-\frac{1}{2}\textbf{e}_x^T \textbf{M}_{\sigma_x(B)} \textbf{e}_x,
 \end{equation}
 where $\textbf{e}_x$ is an all-one column vector.
\end{lemma}
For instance, when no user has been selected, it is evident from Fig.~\ref{systemfigure} that user $v_6$ in the toy example can access four edges from her own collection, namely $\{e_{26}, e_{56}, e_{67}, e_{69}\}$ in $G_1$, as well as two additional edges from her sole friend $v_8$’s collection $\{e_{78}, e_{89}\}$ in $G_1$. Therefore, the initial PoI utility of user $v_6$ is $\phi_{6}(\varnothing) = 6$, which matches the result returned by our approach above in (\ref{toy_initialutility}).

By applying Lemma \ref{lemma_01} to the welfare in (\ref{obj_function}), we obtain the $\Phi(\varnothing)$ in closed form as presented in the following proposition.
\begin{proposition}\label{theorem_02}
When no user is selected yet for GUS in Algorithm \ref{greedyalgforfiniteidea1}, we have $\Phi(\varnothing)=\frac{1}{m}\textbf{e}^T\textbf{C}\textbf{e}-\frac{1}{2m}\sum\limits_{v_x\in V}\textbf{e}^T_x \textbf{M}_{\sigma_x(B)} \textbf{e}_x$,
in which $\textbf{e}$ is an all-one column vector.
\end{proposition}

Given the set $S_{i-1}$ of the first $i-1$ selected users, we now introduce our approach for dynamically computing $\Phi(S_i)$. Suppose the $i$th selected user is node $v_h$, i.e., $s_i=v_h$.
 
  \textbf{Step 1}
($\textbf{B}$ \textit{matrix update}). $\textbf{B}_i$ is updated from $\textbf{B}_{i-1}$ (with $\textbf{B}_0=\textbf{B}$) in the following manner: all the entries on the $h$-th row of matrix $\textbf{B}_{i-1}$ are updated to 1, indicating that the PoI information collected by $v_h$ will be broadcast to $v_i$ once $v_h$ is selected. Meanwhile, the entries on the $h$th column of matrix $\textbf{B}_{i-1}$ remain unchanged, reflecting the fact that  $v_h$ will not increase her own PoI information from being selected. Consequently, $\textbf{B}_i$ may not be symmetric anymore. Based on the matrix $\textbf{B}_i$, the matrix $\textbf{C}_i$ is further updated from $\textbf{C}_{i-1}$ (with $\textbf{C}_0=\textbf{C}$) as follows.

  \textbf{Step 2}$ (\textbf{C}$ \textit{matrix Update}). Given the set $S_i$ of the first $i$ selected users, matrix $\textbf{C}$ can be updated now to $\textbf{C}_i$ by $\textbf{C}_i=\textbf{AB}_i$.
Accordingly, $\textbf{M}_{\sigma_x(\textbf{B})}$ is updated to $\textbf{M}_{\sigma_x(\textbf{B})_i}$.

With updated $\textbf{C}_i$ and $\textbf{M}_{\sigma_x(\textbf{B})_i}$,
the closed-form expression for $\Phi(S_i)$ is derived in the following theorem.
\begin{theorem}\label{thm4.3_for_computing_Phi}
Given the set $S_i$ of selected users, we have
\begin{equation}\label{formulation_phi_s_i_computing}
    \Phi(S_i)=\frac{1}{m}\textbf{e}^T(\textbf{C}_i)\textbf{e}-\frac{1}{2m}\sum\limits_{v_x\in V}\textbf{e}^T_x \textbf{M}_{\sigma_x(\textbf{B}_i)} \textbf{e}_x,
\end{equation}
in which $\textbf{e}$ and $\textbf{e}_x$ are both all-ones column vectors.
\end{theorem}
\subsubsection{Static Crowd-sensing with PoI Preference}
We now adapt our tractable approaches above to the static crowd-sensing scenario with PoI preference, where each user $v_i\in V$ has her own subset $E_1^{(i)}$ of edges of interest. To compute $\phi_i$ in (\ref{obj_function_pre_restricted}) for each user $v_i$, we further update the sensing matrix $\textbf{A}$ to matrix $\textbf{A}^{i}=(a^i_{x,y})_{1\leq x,y\leq \varpi}$, which only includes those edges/roads that are of user $v_i$'s own interest. Now, the entry $(a^i_{x,y})$ at the $x$-th row and the $y$-th column of matrix $\textbf{A}^{i}$ follows: 
\begin{equation}
a^i_{x,y}=\left\{\begin{matrix}
1, & {\rm if\;}x\neq y{\rm \;and\;} (v_x,v_y)\in E_1^{(i)},\\ 
0,& {\rm otherwise}.
\end{matrix}\right.
\end{equation}
By applying a similar matrix-update approach as described above for the static crowd-sensing scenario without PoI preference, we can compute each $\phi_i(S_i)$ and further each $\Phi_i(S_i)$ in the case with PoI preference.

With the above subroutine approaches for computing each  $\Phi(S_i)$ in place, our Algorithm~\ref{greedyalgforfiniteidea1} now is ready to select each hotspot user as specified in its Step~2. 
In the following subsection, we will discuss the approximation guarantees of our Algorithm~\ref{greedyalgforfiniteidea1}.
\subsection{Decent Approximation Guarantee for Algorithm~\ref{greedyalgforfiniteidea1}}\label{sec_approximation_result}

  As a critical step in establishing our approximation results, we first disclose in the following lemma that our objective function $\Phi(V)$ in (\ref{obj_function}) meets some desirable properties, e.g., monotonicity and submodularity. A set function $f(\cdot)$, defined over subsets of a given finite set $V$, is called submodular if  $f(S\cup\{v\})-f(S)\geq f(T\cup\{v\})-f(T)$
holds for all elements $v\in V$ and pairs of sets $S\subseteq V, T\subseteq V$ with $S\subseteq T$. 

\begin{lemma}\label{lemma_02}
Given any sensing graph $G_1=(V\cup V', E_1)$ and social graph $G_2(V,E_2)$, the function
$\Phi(S)$  in (\ref{obj_function}) is non-decreasing and submodular in its input $S\subseteq V$.
\end{lemma}
The following lemma provides another useful component for further establishing our approximation results. 
\begin{lemma}
\label{submodular_supportlemma}
For maximization problems that select $k$ elements one at a time, if the objective function $f(\cdot)$ on subsets of a finite set $V$, is submodular, non-decreasing, and $f(\varnothing)=0$, we can obtain a solution with performance value at least $1-(\frac{k-1}{k})^k$ fraction of the optimal value, which is done by selecting an element that provides the largest marginal increase in the function value.
\end{lemma}
Leveraging Lemmas \ref{lemma_02} and \ref{submodular_supportlemma}, we derive the approximation guarantee for our Algorithm~\ref{greedyalgforfiniteidea1} in the following theorem.
\begin{theorem}\label{ub_OBJ1F}
For the static crowd-sensing problem in (\ref{initial_objective_globally}) with and without PoI preference, Algorithm~\ref{greedyalgforfiniteidea1} with its subroutines approaches discussed in Section~\ref{subsection_computing_phi} guarantees an approximation $1-\frac{m-2}{m}(\frac{k-1}{k})^k>0.632$ of the optimum. It 
runs in polynomial-time with complexity $O(k(m-k)m^{2.2372})$.
\end{theorem}

Note that the approximation guarantee above in Theorem \ref{ub_OBJ1F} exhibits a monotonically decreasing trend as $k$ or $m$ increases. 

%
%

%

\section{Augmentation-Adaptive Algorithms for Mobile Crowd-sensing}\label{sec_resourceaugmented}
  In this section, we simplify the routing complexities of the mobile crowd-sensing scenario by transforming the problem into an optimization alternative. This problem transformation enables us to analyze the problem using the framework established for our fundamental scenario in Section~\ref{sec_fundamentalalg}. 
\subsection{Problem Transformation to Fit Section III's Framework}
  Recall that the path $p_j$ recommended to each selected user $s_j\in S_k$ should meet the following criteria:
\begin{itemize}
    \item Path $p_j$ starts at the location of the user node $s_j$.
    \item The length of path $P_j$ is $n$, i.e.,  $|p_j|=n$.
\end{itemize}
In light of the above two criteria for eligible path recommendations, we can transform our mobile crowd-sensing problem into the new Problem (\ref{new_obj_nhop})-(\ref{new_obj_nhop_constraint3}) through the following steps:
\begin{itemize}
    \item For each user $v_i\in V$, we apply the depth-first-search (DFS) technique \cite{cormen2022introduction} on the sensing graph $G_1$ with $v_i\in V$ being the root, resulting in a DFS tree with its root at $v_i$. Accordingly, we can find all those $n$-hop paths starting from $v_i$ by identifying those $n$-depth edges in the above DFS tree, which runs in $O(|V'|+m+|E_1|)$-time since $|E_1|\leq (|V'|+m)^2$.
By repeating this DFS process for each location node in $V$, we can find all paths that start from a location node and have a length $n$ overall in $O(m(|V'|+m)^2)$-time, which are summarized in the set $\mathcal{P}$.  This path set $\mathcal{P}$ will constitute the search space of the problem transformed.
    \item Find a subset $P\subseteq \mathcal{P}$ of  $k$ paths from $k$ distinct start nodes, such that the objective $\frac{1}{m}\sum\limits_{v_i\in V}|E_1(\{v_i\}\cup N_2(v_i)\cup N_1(P))|$ in  (\ref{obj_function}) is maximized. 
    \item In this way, our problem is equivalently simplified from (\ref{initial_objective_globally}) to the following Problem (\ref{new_obj_nhop})-(\ref{new_obj_nhop_constraint3}):  
\end{itemize}
\begin{eqnarray}
\max\limits_{P\subseteq\mathcal{P}} &\frac{1}{m}\sum\limits_{v_i\in V}|E_1(\{v_i\}\cup N_2(v_i)\cup N_1(P))| \label{new_obj_nhop}\\
{\rm s.t.}& \mathcal{P}\triangleq\{p\subseteq E_1\big | p {\rm \; is\;from\;a\;user\;and}\;|p|=n\},\label{new_obj_nhop_constraint1}\\
%
{}& |P|=k. \label{new_obj_nhop_constraint3}
\end{eqnarray}

\subsection{Augmentation-adaptive Algorithms with Provable Performance Guarantees} 
  Notice that the search space $\mathcal{P}$ of the above Problem (\ref{new_obj_nhop})-(\ref{new_obj_nhop_constraint3}) may include multiple
paths that start from the same user node. This
motivates us to further consider the following question: \textit{can one achieve more welfare efficiency when a user node has multiple users instead}? 

Inspired by \cite{anand2012resource,caragiannis2022truthful,kalyanasundaram2000speed}, we look at an advanced resource-augmented version of the problem, where the high-level insight is that: when the algorithm with an acceptable augmentation outperforms significantly the original solution, it makes sense to invest to invite more users to join at each sensing node.

To begin with, we introduce a new resource augmentation scheme for our problem: given an input instance $I=(G_1,G_2)$ of the two graphs in the mobile crowd-sensing scenario, we create an augmented instance $I_g$ to allow each user node of $G_1$ to have $g\geq 1$ users. Considering the budget $k$ for user selection, a reasonable augmentation factor will also adhere to $g\leq k$. Accordingly, the approximation ratio with augmentation factor $g$ tells the worst-case ratio of an approximation algorithm's welfare on $I_g$ over an optimal solution's welfare on the original $I$, which is now updated from (\ref{def_approratio}) to
\begin{equation*}
    \rho=\inf\limits_{I=(G_1,G_2)}\frac{\Phi(P_k(I_g))}{\Phi(P_k^*(I))}.
\end{equation*}

Next, we propose a new augmentation-adaptive approach to fit any $g$ in Algorithm~\ref{greedyalgforfiniteidea1_augmented}, named Greedy Path Selection (GPS). In general, the GPS first constructs the search space $\mathcal{P}$, and then 
dynamically selects a path that brings the largest marginal welfare, subject to the constraint that the number of paths originating from the same user node does not exceed the augmentation factor $g$. The social matrix now is updated from (\ref{matrixb_eq}) to the $\varpi\times \varpi$ matrix $\textbf{B}^i=(b^i_{x,y})_{1\leq x,y\leq \varpi}$, where
$b^i_{x,y}=\left\{\begin{matrix}
1, & {\rm if\;}x= y{\rm \;or\;} (v_x,v_y)\in  E_2,\\ 
0,& {\rm otherwise}.
\end{matrix}\right.$
\begin{algorithm}[t]
\caption{\textsc{Greedy-Path-Selection} (GPS)}
\label{greedyalgforfiniteidea1_augmented}
\textbf{Input}:  $G_1=(V,E_1)$, $G_2=(V,E_2)$, $n$, $k$ and $g$, $P\leftarrow\varnothing$.\\
\textbf{Output}: $P=\{p_1,...,p_k\}$.
\begin{algorithmic}[1] 
\STATE {\color{black}For each user node $v_i\in V$, compute a tree by applying the depth-first-search (DFS) technique \cite{cormen2022introduction} on the sensing graph $G_1$ with $v_i$ being the root. }
\STATE {\color{black}Construct $\mathcal{P}$ by including all $n$-depth edges of those DFS trees obtained above.}
\FOR{$i\in [k]$}
    \STATE Select $p_i\leftarrow\arg\max\limits_{p\in \mathcal{P}}\left\{{\Phi}(P\cup\{p_i\})-{\Phi}(P)\right\}$.
    \STATE $P\leftarrow P\cup\{p_i\}$.
    \STATE Find in $P$ all those paths that start with the same user node as $p_i$, included in set $P_i$.
    \IF{$|P_i|=g$}
{ \STATE Find in $\mathcal{P}$ all those paths that start with the same user node as $p_i$, included in set $\mathcal{P}_i$.}
        \STATE $\mathcal{P}\leftarrow \mathcal{P}-\mathcal{P}_i$.
    \ENDIF
\ENDFOR
\end{algorithmic}
\end{algorithm}


\begin{figure*}[t]
\centering
\begin{minipage}[t]{0.3\textwidth}
\centering
\includegraphics[width=5.6cm]{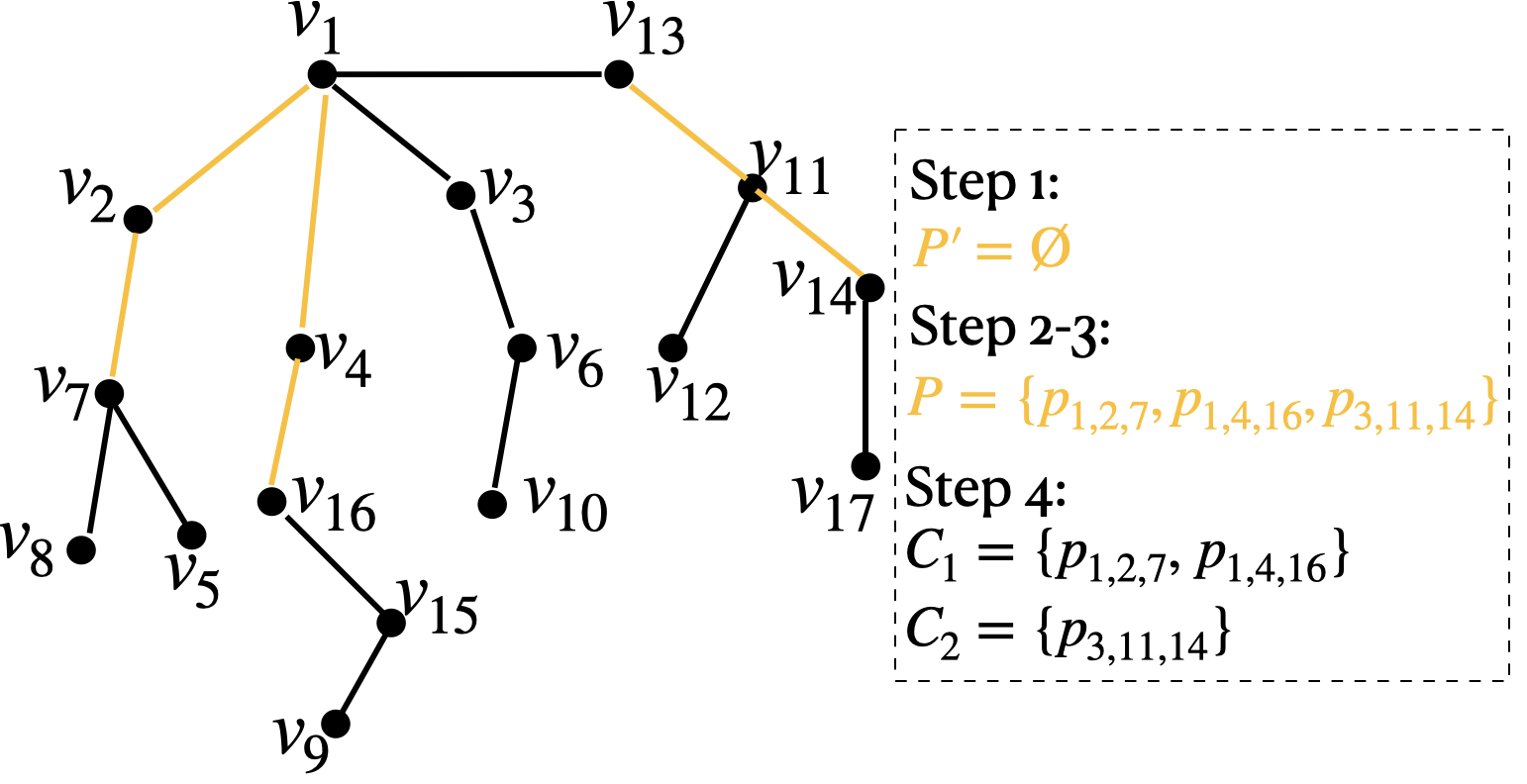}
\end{minipage}
\begin{minipage}[t]{0.30\textwidth}
\centering
\includegraphics[width=5.5cm]{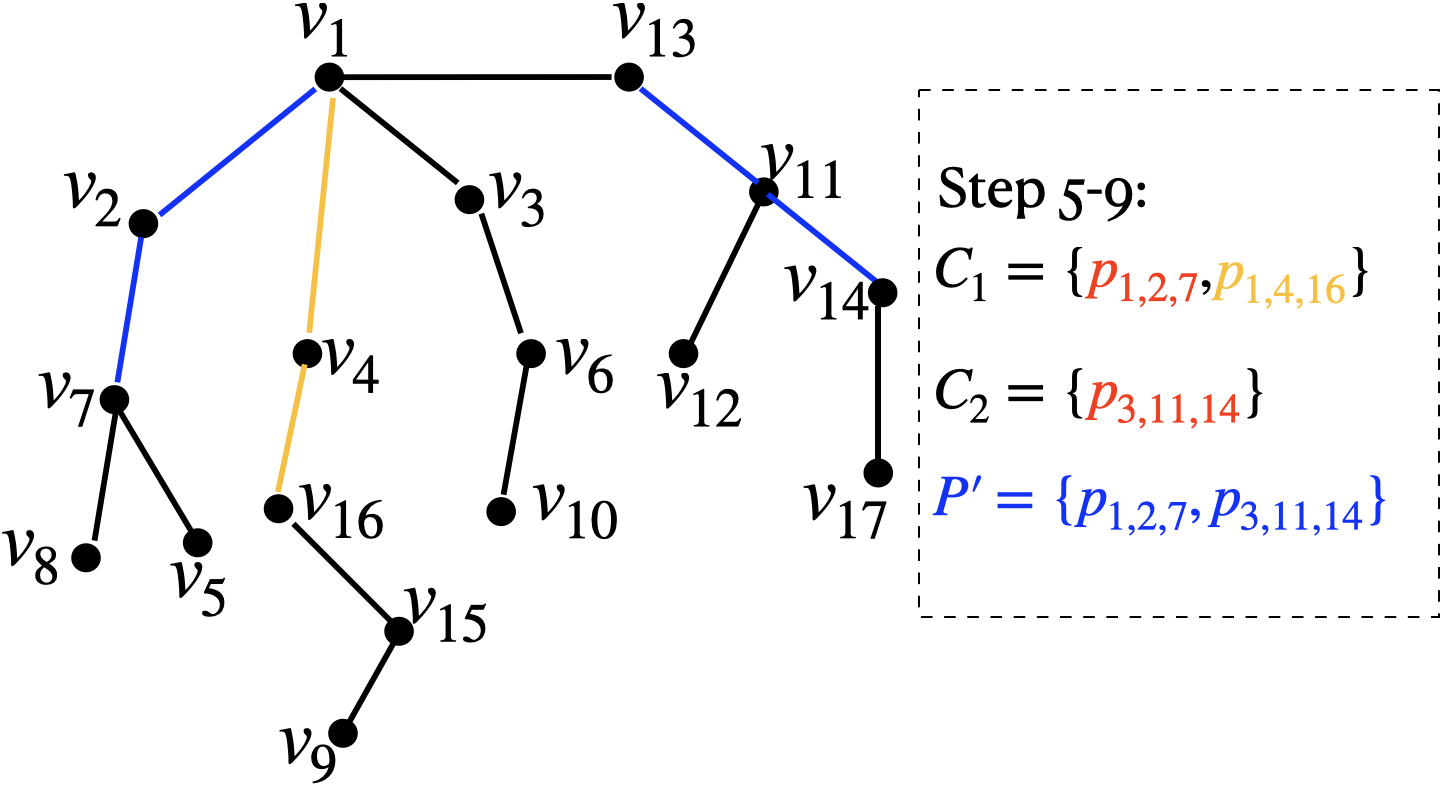}
\end{minipage}
\begin{minipage}[t]{0.3\textwidth}
\centering
\includegraphics[width=6cm]{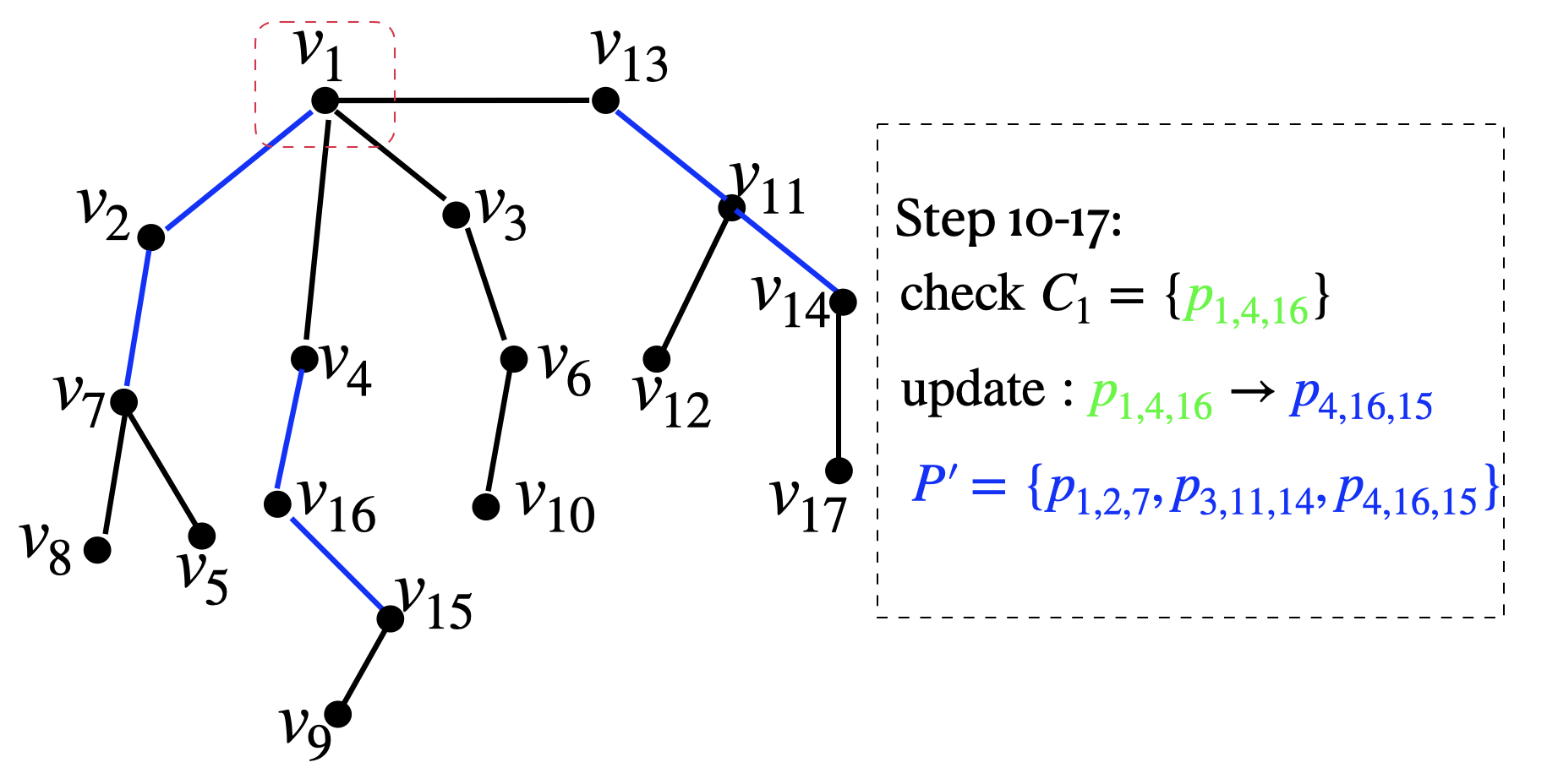}
\end{minipage}
\caption{Illustrating example for Algorithm~\ref{adjusted_greedyhireg_alg} with $k=3$ selected users and $n=2$ hops, where each $p_{xyz}$ denotes, for short, the path $(v_x,v_y,v_z)$.}
\label{figexecutionalg}
\end{figure*}

Since each selected user $v_i$ will collect more PoI information from roads that are incident to the path $p_i$ now, we update in matrix $\textbf{B}_{j-1}$ those row entries (corresponding to nodes in $p_i$) to one. The following theorem summarizes our results for the mobile crowd-sensing scenario.  %
\begin{theorem}\label{theorem_n_hop_forward}
    For the mobile crowd-sensing problem (\ref{new_obj_nhop})-(\ref{new_obj_nhop_constraint3}), Algorithm~\ref{greedyalgforfiniteidea1_augmented} (GPS) with augmentation factor $g\in\{1,...,k\}$ guarantees an approximation ratio of at least $\frac{g}{k}[1-\frac{\varpi-2}{\varpi}(\frac{k-1}{k})^k]$ of the optimum,
    which runs in $O(k(\tau-k)\varpi^{2.372})$-time, where $\tau=|\mathcal{P}|$ denotes the number of all possible $n$-hop paths in $G_1$. 
\end{theorem}


Theorem~\ref{theorem_n_hop_forward} demonstrates that the runtime of Algorithm~\ref{greedyalgforfiniteidea1_augmented} is polynomial in budget $k$ and the node size $\varpi$. As LBSNs typically restrict the number of hops $n$ for maintaining information freshness, Algorithm~\ref{greedyalgforfiniteidea1_augmented} still achieves high time efficiency even in extreme cases where $G_1$ is a complete graph and the factor $\tau\triangleq|\mathcal{P}|$ approaches $\varpi^n$. 

As the augmentation $g$ increase from $1$ to $k$, Theorem~\ref{theorem_n_hop_forward} reveals that the approximation guarantee of Algorithm~\ref{greedyalgforfiniteidea1_augmented} significantly improves by a factor of $k$, increasing from $\frac{1}{k}(1-\frac{1}{e})$ to $1-\frac{1}{e}> 0.632$. This validates the motivation to augment resources by inviting more users to dense urban networks, thereby enhancing the efficiency of PoI sensing and sharing. 

When no augmentation is applied or $g=1$, our Algorithm~\ref{greedyalgforfiniteidea1_augmented}
only achieve an
approximation guarantee $\frac{1}{k}(1-\frac{1}{e})$, which falls short of being satisfactory. Nevertheless, we surprisingly find in the following subsection that even without augmentation, Algorithm~\ref{greedyalgforfiniteidea1_augmented} can still be enhanced to further improve its approximation guarantee significantly by $k$ times as long as each node of the sensing graph $G_1$ has a user.  

\subsection{Enhanced Algorithm for $G_1$ with Only User Nodes}
  Notice that different sets of $k$ selected paths in the above Problem (\ref{new_obj_nhop})-(\ref{new_obj_nhop_constraint3}) can yield the same welfare contributions, as long as they visit the same set of nodes in the sensing graph $G_1$. This inspires us to look at a typical and interesting case where each node in the sensing graph is associated with only one user, i.e., $G_1=(V,E_1)$. Accordingly, we 
propose a more efficient Algorithm~\ref{adjusted_greedyhireg_alg}, which leverages our Algorithm~\ref{greedyalgforfiniteidea1_augmented} (GPS) as a subroutine. To differentiate from the solution $P$ output by Algorithm \ref{greedyalgforfiniteidea1_augmented} with augmentation $k$, we denote the solution by Algorithm~\ref{adjusted_greedyhireg_alg} as $P'$.
There are mainly three steps in Algorithm~\ref{adjusted_greedyhireg_alg}:
\begin{itemize}
    \item \textit{First}, Algorithm~\ref{adjusted_greedyhireg_alg} partitions paths into several classes according to their start nodes, which are included in the class set $\mathcal{C}$;
    \item \textit{Then}, from each class $C_i\in\mathcal{C}$, Algorithm~\ref{adjusted_greedyhireg_alg} selects the first path output by Algorithm~\ref{greedyalgforfiniteidea1_augmented}, which is included in $P'$;
    \item \textit{Finally}, for each other path in $P$ but not in $P'$, Algorithm~\ref{adjusted_greedyhireg_alg} constructs a new $n$-hop path starting at a different node from any start node of a path in the current set $P'$, and then include this new path in $P'$. 
\end{itemize}

Fig. \ref{figexecutionalg} shows an illustrative example with $k=3$ selected users to broadcast and $n=2$ hops in each path for explaining the execution of Algorithm~\ref{adjusted_greedyhireg_alg}. Here, suppose that Algorithm~\ref{greedyalgforfiniteidea1_augmented} with augmentation factor $k$ returns a solution $P =\{(v_1,v_2,v_7),(v_1,v_4,v_{16}),(v_3,v_{11},v_{14})\}$. In the left-most sub-figure of Fig. \ref{figexecutionalg}, Algorithm~\ref{adjusted_greedyhireg_alg} partitions paths in $P$ into two classes $\{(v_1,v_2,v_7),(v_1,v_4,v_{16})\}$ and $\{(v_3,v_{11},v_{14})\}$.  As shown in blue in the middle sub-figure of Fig. \ref{figexecutionalg}, Algorithm~\ref{adjusted_greedyhireg_alg} then selects to the first path $(v_1,v_2,v_7)$ and $(v_3,v_{11},v_{14})$ of the two partitioned classes, respectively. Finally, in the right-most sub-figure of Fig. \ref{figexecutionalg}, Algorithm~\ref{adjusted_greedyhireg_alg} checks the remaining path $(v_1,v_4,v_{16})$, along which node $v_4$ is the first node different from any start nodes of previously selected paths. This is because the predecessor of $V_4$ on the path, i.e., $v_1$, is the start node of the previously selected path $(v_1,v_2,v_7)$. Accordingly,  
Algorithm~\ref{adjusted_greedyhireg_alg} updates $(v_1, v_4,v_{16})$ to  $(v_4,v_{16},v_{15})$ by moving one hop forward accordingly.

\begin{algorithm}[t]
\caption{\textsc{Adjusted-GPS}}\label{adjusted_greedyhireg_alg}
\textbf{Input}: $I=(G_1,G_2)$, $n$, $k$, $P'=\varnothing$.\\
\textbf{Output}: $P'=\{p'_1,...,p'_k\}$.
\begin{algorithmic}[1] 
\STATE Construct an augmented instance $I_k$ from $I$ by an augmentation factor $k$.
\STATE Obtain a solution $P$ (of paths) for the augmented instance  $I_k$
by Algorithm \ref{greedyalgforfiniteidea1_augmented} with augmentation $k$.
\STATE Partition paths in $P$ into distinct classes according to their start nodes, resulting in a class set denoted by $\mathcal{C}$.
\FOR{each $C_i\in\mathcal{C}$}
\STATE  Denote $C_i(1)$ as the first path selected by Algorithm~\ref{greedyalgforfiniteidea1_augmented}.
\STATE Update $P'=P'+\{C_i(1)\}$.
\STATE Update $C_i\leftarrow C_i-\{C_i(1)\}$.
\ENDFOR
\FOR{each $C_i\in\mathcal{C}$}
\FOR{each $p\in C_i$}
\IF{there exist a node in $p$ that is not a start node of any path selected in set $P'$}
\STATE Denote $p(1')$ as the first node in $p$ that is not a start node of any path selected in set $P'$, where $1'$ indexes the node's order in $p$.
\STATE Extend the sub-path of $p$ that starts from node $p(1')$ to a new $n$-size path, include the new path to $P'$.
\ENDIF
\ENDFOR
\ENDFOR
\end{algorithmic}
\end{algorithm}
The following theorem validates that even without augmentation, Algorithm~\ref{adjusted_greedyhireg_alg} still improves the approximation guarantee up to a factor of $k$ for user-node-only sensing graphs.

\begin{theorem}\label{improvedratio_nhop}
    For the mobile crowd-sensing scenario in a sensing graph $G_1(V,E_1(V))$ of user nodes only, Algorithm~\ref{adjusted_greedyhireg_alg} using Algorithm~\ref{greedyalgforfiniteidea1_augmented} as a subroutine runs in $O(k(\tau-k)\varpi^{2.372})$-time and guarantees an approximation ratio of at least $1-\frac{\varpi-2}{\varpi}(\frac{k-1}{k})^k$, where $\tau$ and $\varpi$ denote the number of $n$-hop paths of interest and the number of nodes in the sensing graph $G_1$, respectively.
\end{theorem}

\section{Experiments with Real Dataset Across Different Network Topologies}\label{section_experiment}
  Now, we examine the empirical performances of our algorithms via simulations, where we apply a publicly available dataset of check-ins collected by Gowalla users \cite{Gowalladata,cho2011friendship}. {\color{black}In the simulation, we select two different parts of sensing areas in northern San Francisco, which feature two typical network topologies. On one hand, our first sensing area focuses on the region with latitude ranging from $37^{\circ}46' 20'' N$ to $37^{\circ}47' N$ and longitude ranging from $122^{\circ}26' 30'' W$ to $122^{\circ}25' 30'' W$, which exhibits a grid network topology as depicted in Fig.~\ref{urban_sensing_fig}. On the other hand, our second sensing area showcases a mesh network topology as illustrated in Fig.~\ref{new_urban_sensing_fig}, spanning from latitude $37^{\circ}47' 40'' N$ to $37^{\circ}48' 20''N$ and from longitude $122^{\circ}27' 55'' W$ to $122^{\circ}27' 40'' W$.

\begin{figure}[h]
   \begin{minipage}{0.23\textwidth}
   \centering
  \includegraphics[width=5.5cm]{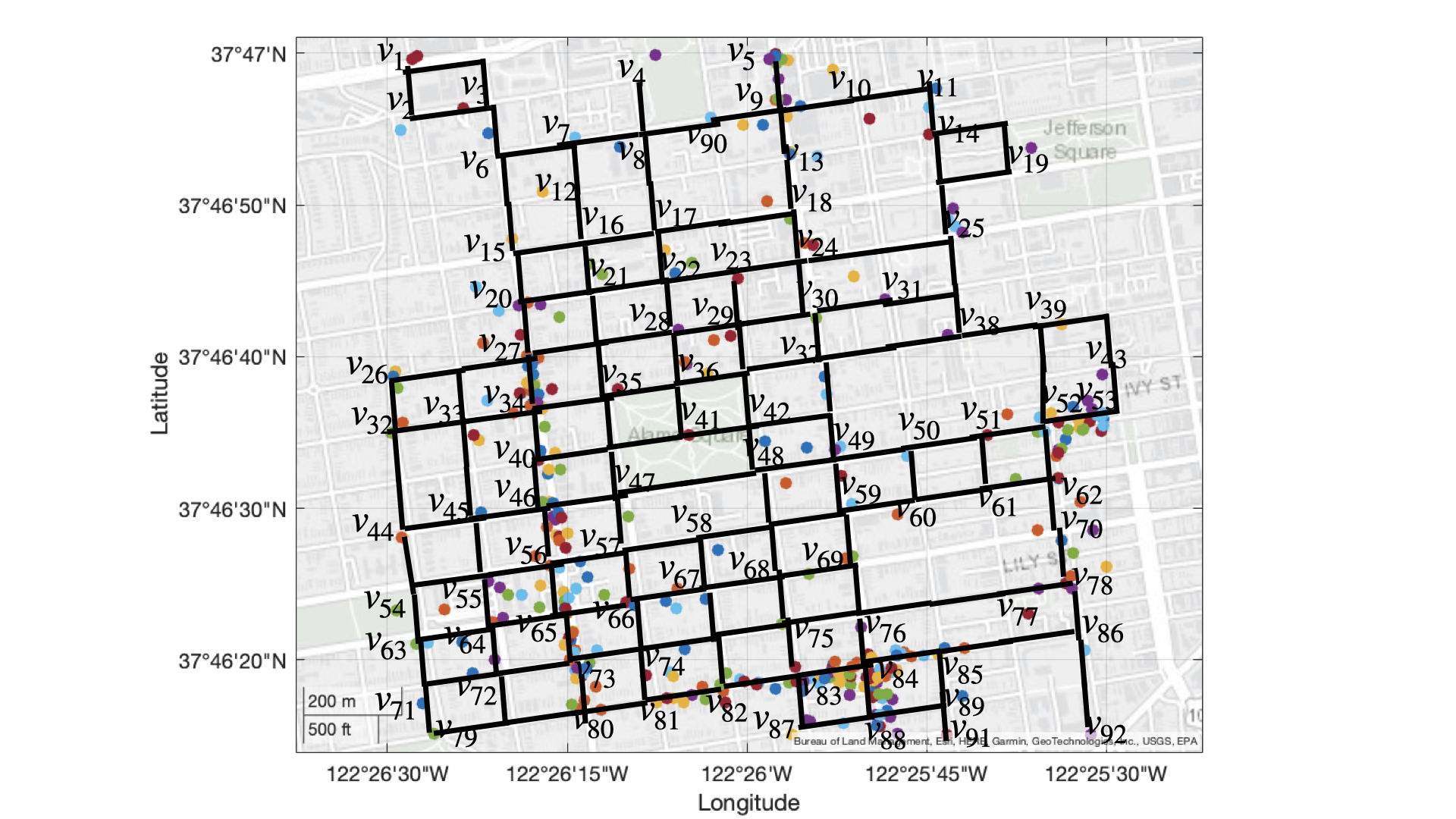}
     \caption{Grid sensing network \cite{Gowalladata}.}\label{urban_sensing_fig}
   \end{minipage}
   \begin{minipage}{0.24\textwidth}
   \centering
\includegraphics[width=4cm]{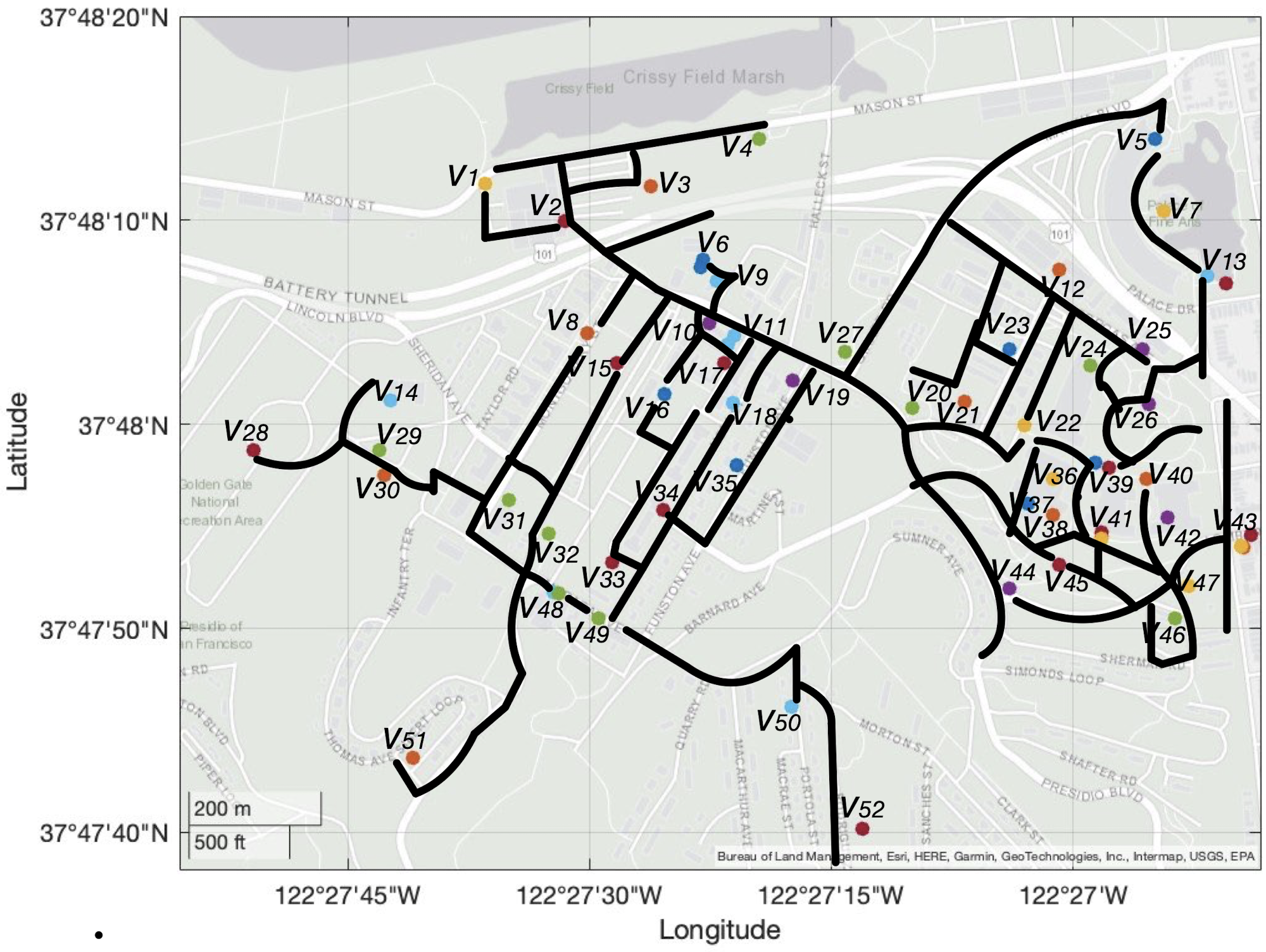}
     \caption{Mesh sensing network \cite{Gowalladata}.}\label{new_urban_sensing_fig}
   \end{minipage}
\end{figure}

Accordingly, we extract the dataset to include the check-in locations falling within the target sensing areas, resulting in a total of 2,176 user check-ins in Fig.~\ref{urban_sensing_fig} for the grid sensing network and a total of 1,483 user check-ins in Fig.~\ref{new_urban_sensing_fig} for the mesh sensing network, respectively. Furthermore, we cluster the 2,176 and 1,483 extracted check-ins into 92 and 52 distinctive locations, respectively, serving as the input user nodes for our sensing graphs depicted in Fig.~\ref{urban_sensing_fig} and Fig.~\ref{new_urban_sensing_fig}.
To input the social graph in the experiments, we follow a commonly adopted Gaussian distribution \cite{tripathi2022optimizing} to randomly construct social connections among 92 and 52 users for the two sensing networks in Fig.~\ref{urban_sensing_fig} and 
Fig.~\ref{new_urban_sensing_fig}, respectively.}

To begin with, we evaluate our algorithm GUS under the static crowd-sensing scenario. Since finding an optimal solution $\Phi(S_k^*)$ is NP-hard, we compare the performance of algorithm GUS with an upper bound ${\rm UB}_1\triangleq\Phi(\varnothing)+ |SCP(G_1,k)|$ for the optimum $\Phi(S^*_k)$. Here, $|SCP(G_1,k)|$ represents the maximum number of edges incident to a group of $k$ users in the sensing graph $G_1$, which can be obtained from the set cover solution with budget $k$ \cite{feige1998threshold}. In other words, $|SCP(G_1,k)|$ reflects the largest possible amount of PoI information that a subset of $k$ selected users could collect. 
Accordingly, we use the following benchmark ratio $\underline{\rho_1}$ to evaluate the performance of our algorithms and benchmarks against the upper bound $UB_1$ for various budget values $k$ in the simulation. 
\begin{equation}\label{expediment_ratio}
    \underline{\rho_1} \triangleq \frac{\Phi(S_k)}{{\rm UB}_1}< \frac{\Phi(S_k)}{\Phi(S_k^*)}.
\end{equation}
This ratio provides a lower bound on the efficiency ratio of the algorithm’s actual performance, enabling an indirect comparison of our algorithm GUS' empirical performances with the optimal solution under various budgets $k$. 

In Fig.~\ref{real_panoramic_35_fig} (for grid sensing network) and Fig.~\ref{new_real_panoramic_35_fig} (for mesh sensing network), we evaluate our Algorithm~\ref{greedyalgforfiniteidea1}'s performances against three benchmarks, which is measured by the performance ratio $\underline{\rho_1}$ defined in (\ref{expediment_ratio}):
\begin{itemize}
    \item \textit{First}, the blue curve represents our proved approximation guarantee $1-\frac{m-2}{m}(1-\frac{1}{k})^k$ in Theorem~\ref{ub_OBJ1F} which reflects the worst-case theoretical  
 guarantee for GUS. 
    \item \textit{Second,} the red curve represents the $k$ user nodes output 
 by the classic set cover solution $SCP(G_1,k)$ \cite{feige1998threshold}.
    \item \textit{Finally,} the black curve indicates the case without social broadcasting, in which no user is selected and each user $v_i$ can only access data from her own PoI collection in $E_1(v_i)$ and her friends' collection in $E_1(N_2(v_1))$. 
\end{itemize}
\begin{figure*}
\begin{minipage}[t]{0.245\textwidth}
    \centering
\includegraphics[width=4.5cm]{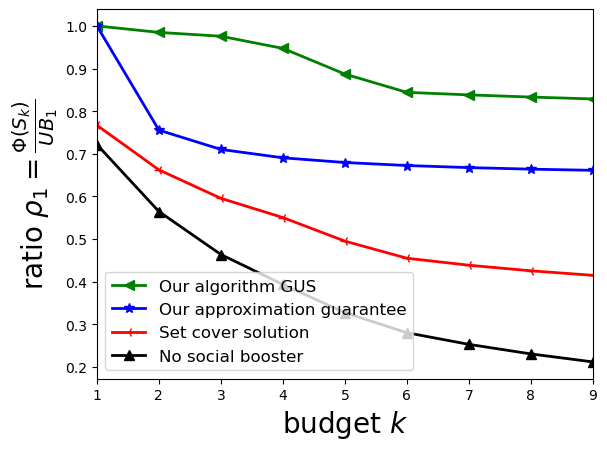}
\caption{Under grid sensing network: performance ratio $\underline{\rho_1}$ in (\ref{expediment_ratio}) of GUS (i.e., Algorithm~\ref{greedyalgforfiniteidea1}) and its benchmarks as compared to the optimum.}
    \label{real_panoramic_35_fig}
\end{minipage}
\begin{minipage}[t]{0.245\textwidth}
    \centering
\includegraphics[width=4.5cm]{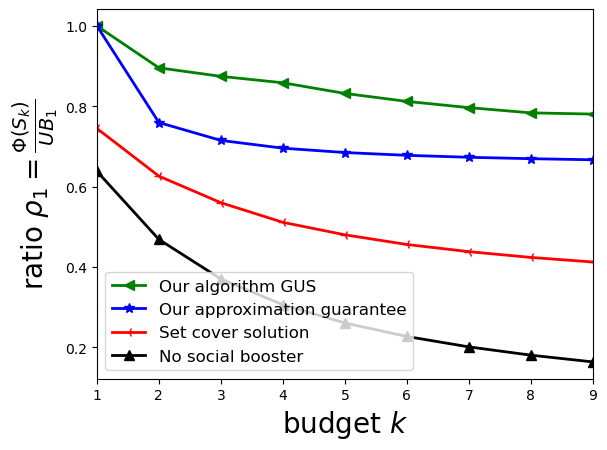}
\caption{Under mesh sensing network: performance ratio $\underline{\rho_1}$ in (\ref{expediment_ratio}) of GUS (i.e., Algorithm~\ref{greedyalgforfiniteidea1}) and its benchmarks as compared to the optimum.}
    \label{new_real_panoramic_35_fig}
    \end{minipage}
\begin{minipage}[t]{0.245\textwidth}
    \centering
\includegraphics[width=4.5cm]{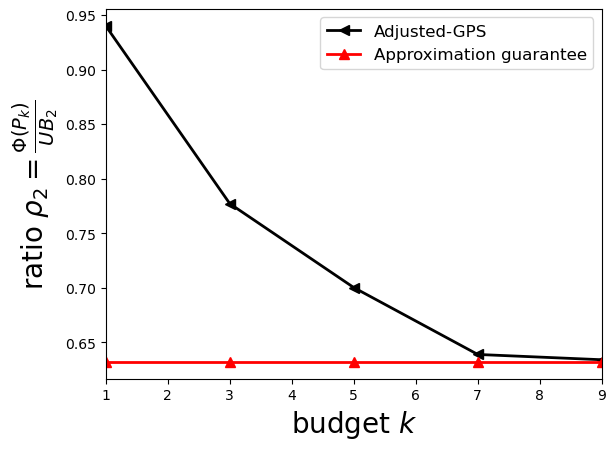}
\caption{Under grid sensing network: performance ratio $\underline{\rho_2}$ in (\ref{expediment_ratio_new}) of Algorithm~\ref{adjusted_greedyhireg_alg} and its approximation guarantee as compared to the optimum.}
\label{real_n_hop_figure}
\end{minipage}
\begin{minipage}[t]{0.245\textwidth}    \centering
\includegraphics[width=4.5cm]{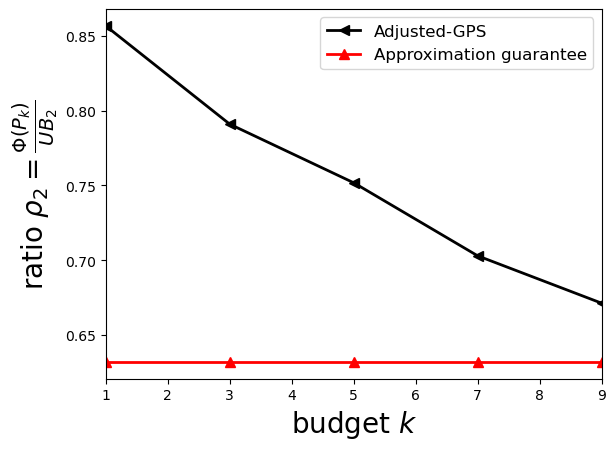}
\caption{Under mesh sensing network: performance ratio $\underline{\rho_2}$ in (\ref{expediment_ratio_new}) of Algorithm~\ref{adjusted_greedyhireg_alg} and its approximation guarantee as compared to the optimum.}
\label{new_real_n_hop_figure}
\end{minipage}
\end{figure*}

Notably, both Fig.~\ref{real_panoramic_35_fig} and Fig.~\ref{new_real_panoramic_35_fig} demonstrate that our algorithm GUS achieves near-optimal performances (close to 1), particularly when the budget $k$ is small. Besides, all four curves for our algorithm and its benchmarks in both figures decrease as the budget $k$ increases. This decrease is due to the fact that the maximum number of edges incident to a group of $k$ users in the sensing graph, represented by $|SCP(G_1,k)|$ in the denominator of the performance ratio $\underline{\rho_1}$, grows significantly with $k$. Interestingly, our empirical results show that GUS' efficiency under grid (resp. mesh) sensing network converges to around 82\% (resp. 80\%) of the optimum as the budget $k$ increases.

This is a significant improvement over the worst-case performance guarantee of $63.2\%$ stated in Theorem~\ref{ub_OBJ1F}. In contrast, solutions derived from the set cover problem remarkably underperform compared to GUS across both sensing network topologies. Specifically, as the budget $k$ increases,  set cover solutions achieve only about half the efficiency of GUS.
Compared to the case without social enhancement or broadcast, as indicated by the black curve, our GUS algorithm can substantially improve the average PoI information and welfare for all users even with just a few users being selected, which is regardless of the network topology. This highlights the great benefit of our social-enhanced sharing approach to selecting and broadcasting hotspot users.

Next, we evaluate the empirical performance of our algorithm \textsc{Adjusted}-GPS for the mobile crowd-sensing scenario under the grid and the mesh sensing networks, respectively. Since the mobile crowd-sensing problem is NP-hard, we use an alternative upper bound, ${\rm UB}_2\triangleq\Phi(\varnothing)+ |SCP(G_1,2k)|$, to approximate the optimum $\Phi(S_k^*)$.
This upper bound ${\rm UB}_2$ is obtained by releasing Constraint (\ref{new_obj_nhop_constraint1}) from  Problem (\ref{new_obj_nhop})-(\ref{new_obj_nhop_constraint3}), and by selecting $n\cdot k$ user nodes for each $k$. 

For different budgets $k$, we evaluate our algorithm \textsc{Adjusted}-GPS using the following new ratio: 
\begin{equation}\label{expediment_ratio_new}
\underline{\rho_2}=\frac{\Phi(P_k)}{{\rm UB}_2}\leq \frac{\Phi(P_k)}{\Phi(P^*_k)}.
\end{equation} 
To the best of our knowledge, no existing benchmark has employed the resource augmentation technique. Therefore, besides the optimum's upper bound $UB_2$, we compare the performance ratio of our algorithm \textsc{Adjusted}-GPS using $\underline{\rho_2}$ from (\ref{expediment_ratio_new}), and contrast it with the theoretical approximation guarantee provided in Theorem~\ref{theorem_n_hop_forward}. Our simulation results for \textsc{Adjusted}-GPS are summarized in Fig.~\ref{real_n_hop_figure} and Fig.~\ref{new_real_n_hop_figure} for the grid and mesh sensing networks, respectively. Both figures show that our \textsc{Adjusted}-GPS algorithm achieves near-optimal performance, particularly when budget $k$ is small. 
Recall that the ratio $\underline{\rho_2}$ applies an upper bound ${\rm UB}_2$ of the optimum as its denominator instead of the actual optimum. As the budget $k$ increases, Fig.~\ref{new_real_n_hop_figure} demonstrates that our \textsc{Adjusted}-GPS algorithm experiences a gradual decrease in its efficiency. As the user/node size scales up  (e.g., from 52 in Fig.~\ref{new_real_n_hop_figure} to 92 in Fig.~\ref{real_n_hop_figure}), we observe an acceleration in the efficiency decline (as the budget $k$ increases) of our \textsc{Adjusted}-GPS. Nonetheless, our \textsc{Adjusted}-GPS converges to a steady efficiency level of approximately 63.2\% of
the optimal solution, which matches our theoretical approximation as proved in Theorem~\ref{theorem_n_hop_forward}. 

\section{Concluding Remarks}\label{section_conclusion}
  This work studies the social-enhanced  PoI sharing problem for maximizing the average amount of useful PoI information to users, which involves the interplay between urban sensing and online social networks. We prove that the problem is NP-hard and renders existing approximation solutions infeasible. By transforming the involved PoI-sharing process across the two networks into matrix computations, we successfully derive the closed-form objective. For static users who collect PoIs only within the vicinity, we leverage the desired properties of the objective function to propose a polynomial-time algorithm with a good approximation guarantee. Furthermore, we apply a new resource-augmentation technique to address a generalized scenario where selected users can move to sense and share additional PoIs, which still guarantees decent approximations. Finally, we conduct simulations across different network topologies, which well corroborate our theoretical results. 

There are several promising directions remaining for future exploration. One natural extension is to design adaptive or online algorithms that can handle dynamic changes in network topology over time. Additionally, it would also be interesting to incentivize users, particularly through pricing mechanisms, to alter their routes to collect more PoIs.


\textbf{Acknowledgement.} This work is supported by the SUTD Kickstarter Initiative (SKI) Grant with no. SKI 2021\_04\_07 and the Joint SMU-SUTD Grant with no. 22-LKCSB-SMU-053.

\nocite{langley00}
\bibliographystyle{IEEEtran}
\bibliography{mybibiligraph}
\begin{IEEEbiography}[{\includegraphics[width=1in,height=1.23in,clip,keepaspectratio]{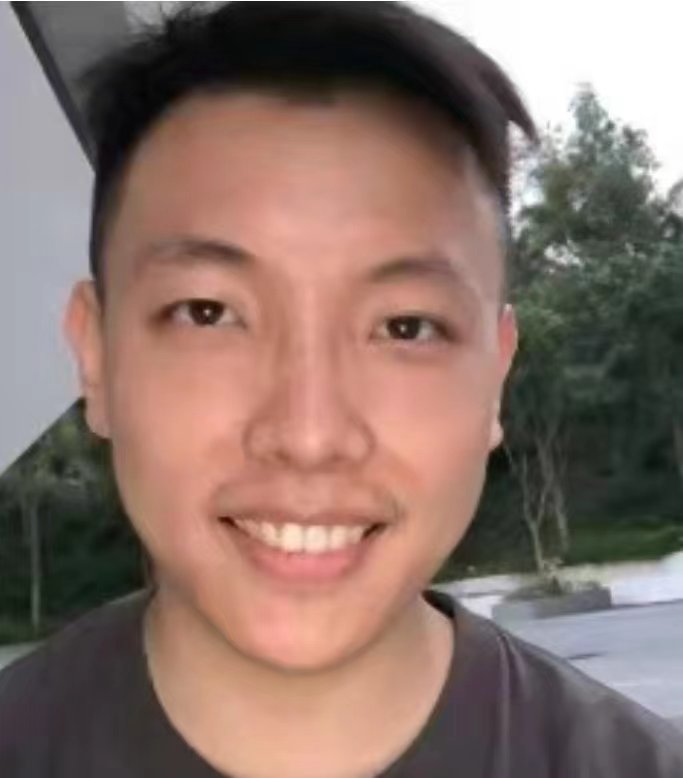}}]{Songhua Li} received
the Ph.D. degree from The City University
of Hong Kong in 2022. He is now a postdoctoral research fellow at 
the Singapore University of Technology and Design (SUTD), where he was a Visiting Scholar in 2019. His research interests span approximation and online algorithm design and analysis, combinatorial optimization, and algorithmic game theory in the context of sharing economy networks.
\end{IEEEbiography}
\begin{IEEEbiography}[{\includegraphics[width=1in,height=1.25in,clip,keepaspectratio]{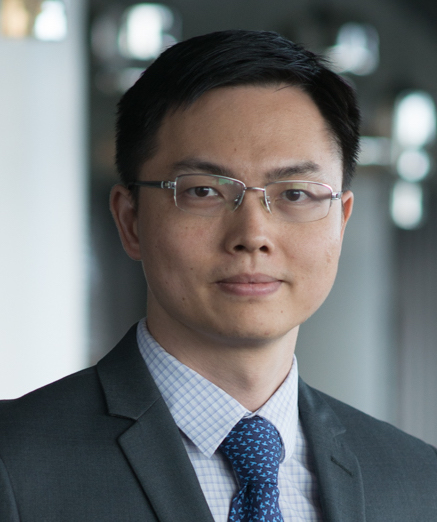}}]{Lingjie Duan} (S’09-M’12-SM’17) received the Ph.D. degree from The Chinese University of Hong Kong in 2012. He is an Associate Professor at the Singapore University of Technology and Design (SUTD) and is an Associate Head of Pillar (AHOP) of Engineering Systems and Design. In 2011, he was a Visiting Scholar at the University of California at Berkeley, Berkeley, CA, USA. His research interests include network economics and game theory, network security and privacy, energy harvesting wireless communications, and mobile crowdsourcing. He is an Associate Editor of IEEE/ACM Transactions on Networking and IEEE Transactions on Mobile Computing. He was an Editor of IEEE Transactions on Wireless Communications and IEEE Communications Surveys and Tutorials. He also served as a Guest Editor of the IEEE Journal on Selected Areas in Communications Special Issue on Human-in-the-Loop Mobile Networks, as well as IEEE Wireless Communications Magazine. He is a General Chair of WiOpt 2023 Conference and is a regular TPC member of some other top conferences (e.g., INFOCOM, MobiHoc, SECON). He received the SUTD Excellence in Research Award in 2016 and the 10th IEEE ComSoc Asia-Pacific Outstanding Young Researcher Award in 2015. 
\end{IEEEbiography}

\newpage
\appendix
\section{Omitted Discussions}\label{appendix_omitteddiscussions}
\subsection{Relation between Our Social-enhanced PoI Sharing Problem And The Set Cover Problem (SCP)}
We note that, despite a relation of our problem to the set cover problem \cite{hochbaum1982approximation,karakostas2009better}, one cannot apply set cover solutions to our problem. 
\begin{remark}[A relation between our problem and the set cover problem (SCP)]\label{remark_conversionproblem}
Through the following four steps, we establish a mapping $G_1(\overline{V},E_1)\rightarrow SCP(G_1,k)$ to convert our problem to a set cover problem with budget $k$:
\begin{itemize}
    \item \textbf{Step 1:} Each edge $e\in E_1$ is mapped to a ground element in SCP. Accordingly, our edge set $E_1$ is mapped to the ground set of $SCP(G_1,k)$. 
    \item \textbf{Step 2:} Each $E_1(v)$  is mapped to a set in the collection of $SCP(G_1,k)$, where $v\in V$.\footnote{This is because we only select user nodes from set $V$. Consequently, nodes in $V'$ are not considered in SCP either.}
    \item \textbf{Step 3:}  $\{E(v)|v\in V\}$ is mapped to the collection of sets in $SCP(G_1,k)$.
    \item \textbf{Step 4:} $SCP(G_1,k)$ looks for a sub-collection of $k$ sets whose union covers the most ground elements.
\end{itemize}
As such, an SCP solution yields a sub-collection of $k$ edge subsets of $\{E(v)|v\in V\}$ that maximizes the number of edges incident to nodes of the sub-collection, returning us those $k$ user nodes in the sub-collection to select. However, as discussed earlier in Section \ref{section_introduction}, such selected user nodes suffer from huge efficiency loss in our problem since most of their incident edges could have been already shared in a relatively well-connected social network. For example, in Fig.~\ref{systemfigure}, assume that user node $6$ happens to have a social connection with each of the other nodes and the hiring budget is reduced to $1$.
It is obvious that broadcasting node 6's PoI collection will not contribute any new content to other users, which is because node 6's PoI collection has already been shared with all other users from the social network. Hence, node $6$ will not be selected by an optimal solution to our problem. However,  node $6$ will be selected by an optimal SCP solution to cover the maximum number of edges in the sensing graph.
\end{remark}
\subsection{Equivalence Relation between The Optimization and Decision Versions of The Set Cover Problem}
 In Fig.~\ref{fig02}, we present the equivalence relation between the decision version and the well-known optimization version of the set cover problem. 
\begin{figure}[h]
    \centering
    \includegraphics[width=9cm]{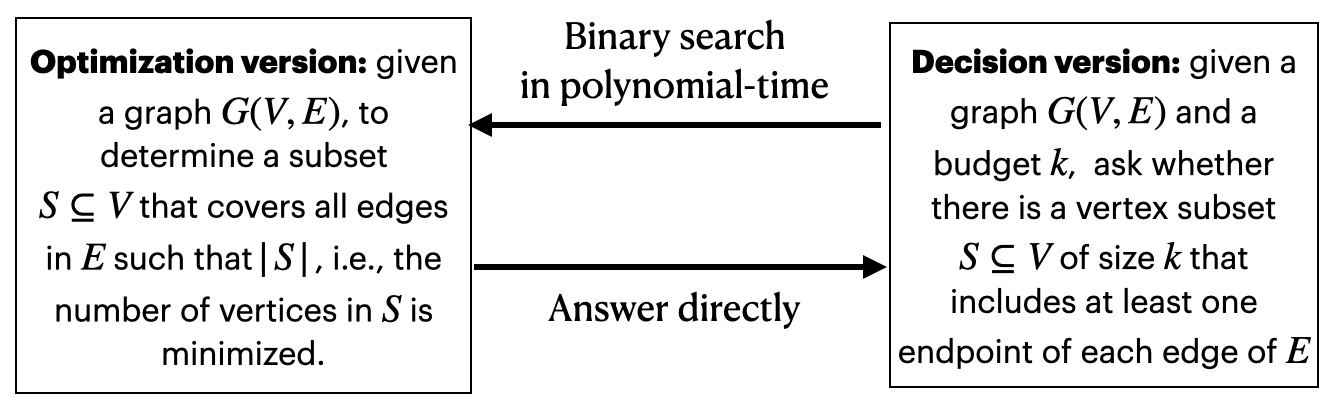}
    \caption{The equivalent relation between the decision and optimization versions of the set cover problem.}
    \label{fig02}
\end{figure}
\subsection{Algorithm~\ref{compputing_phi_varnothing}: Pseudocode for Computing $\Phi(\varnothing)$}
\begin{algorithm}[h]
\caption{\textsc{Computing $\Phi(\varnothing)$ before any hiring}}\label{compputing_phi_varnothing}
\textbf{Input}:  sensing graph $G_1$ and social graph $G_2$.\\
\begin{algorithmic}[1] 
\STATE Construct $\textbf{A}=(a_{i,j})_{1\leq i,j\leq m}$, where 

$$a_{i,j}=\left\{\begin{matrix}
1, & {\rm if\;}i\neq j{\rm \;and\;} (v_i,v_j)\in E_1,\\ 
0,& {\rm otherwise}.
\end{matrix}\right.$$
\STATE Construct $\textbf{B}=(b_{i,j})_{1\leq i,j\leq m}$, where 
$$b_{i,j}=\left\{\begin{matrix}
0, & {\rm if\;}i\neq j{\rm \;and\;} (v_i,v_j)\notin  E_2,\\ 
1,& {\rm otherwise}.
\end{matrix}\right.$$
\STATE Construct $\textbf{C}\triangleq \textbf{AB}$.
\FOR{$v_x\in V$}
\STATE Construct \textbf{$M_{\sigma_x(B)}$} by selecting those columns and rows in matrix $\textbf{A}$ that correspond to either user $v_x$ or $v_x$'s friends in set $N_2(v_x)$.
\STATE $\phi_x(\varnothing)\leftarrow (\sum\limits_{1\leq j\leq m}c_{j,x}-\frac{1}{2}\textbf{e}^T_x \textbf{M}_{\sigma_x(B)}\textbf{e}_x)$.
\ENDFOR
\STATE \textbf{Output} $\Phi(\varnothing)\leftarrow \frac{1}{m}\sum\limits_{1\leq x\leq m}\phi_x(\varnothing)$.
\end{algorithmic}
\end{algorithm}
\subsection{Algorithm~\ref{compputing_phi_si}: Pseudocode for Computing $\Phi(S_i)$ for Static Crowd-sensing without PoI Preference}
\begin{algorithm}[h]
\caption{\textsc{Computing $\Phi(S_i)$ for Static Crowd-sensing without PoI Preference}}\label{compputing_phi_si}
\textbf{Input}:  $G_1=(V,E_1)$, $G_2=(V,E_2)$, $S_i=\{s_1,...,s_i\}$.\\
\begin{algorithmic}[1] 
\STATE Construct matrices $\textbf{A}$ and $\textbf{B}$ by Algorithm \ref{compputing_phi_varnothing}, $\textbf{B}_0\leftarrow \textbf{B}$.
\FOR{$s_j\in(s_1,...,s_i)$}
\STATE Find the index of $s_j$ in set $V$, denoted as $j'$. 
\STATE Update the $j'$-th row entries of matrix $\textbf{B}_{j-1}$ as one.
\ENDFOR
\STATE Compute $\textbf{C}_i\leftarrow \textbf{AB}_i$.
\STATE \textbf{Output} $ \Phi(S_i)\leftarrow\frac{1}{m}\textbf{e}^T\textbf{C}_i\textbf{e}-\frac{1}{2m}\sum\limits_{v_x\in V}\textbf{e}^T_x \textbf{M}_{\sigma_x(B_i)} \textbf{e}_x$.
\end{algorithmic}
\end{algorithm}
\subsection{Algorithm~\ref{compputing_phi_varnothing_restricted}: Pseudocode for Computing $\Phi(S_i)$ for Static Crowd-sensing with PoI Preference}
\begin{algorithm}[h]
\caption{\textsc{Computing $\Phi(S_y)$ for static crowd-sensing with PoI preference}}\label{compputing_phi_varnothing_restricted}
\textbf{Input}:  $G_1$, $G_2$, $m$, $\{E_1^{(1)},...,E_1^{(m)}\}$, $S_y=\{s_1,...,s_y\}$.\\
\begin{algorithmic}[1] 
\STATE Construct $\textbf{B}_0=(b_{i,j})_{1\leq i,j\leq m}$, where 
$$b_{i,j}=\left\{\begin{matrix}
0, & {\rm if\;}i\neq j{\rm \;and\;} (v_i,v_j)\notin  E_2\\ 
1,& {\rm otherwise}
\end{matrix}\right.$$
\FOR{$s_i\in(s_1,...,s_y)$}
\STATE Find the index of $s_i$ in set $V$, denoted as $i'$. 
\STATE Construct matrix $\textbf{B}_{i}$ by updating in $\textbf{B}_{i-1}$ all the $i'$-th row entries to one.
\ENDFOR
\FOR{$v_x\in V$}
\STATE Construct $\textbf{A}^{x}=(a^x_{i,j})_{1\leq i,j\leq m}$,
\begin{equation*}
a^x_{i,j}=\left\{\begin{matrix}
1, & {\rm if\;}i\neq j{\rm \;and\;} (v_i,v_j)\in E_1^{(x)}\\ 
0,& {\rm otherwise}
\end{matrix}\right.
\end{equation*}
\STATE Compute $\textbf{C}^{x}\triangleq \textbf{A}^x \textbf{B}_y$.
\STATE Compute $\textbf{M}_{\sigma_x(B^x)}$.
\STATE $\phi_x(S_y)\leftarrow (\sum\limits_{1\leq j\leq m}c^i_{j,x}-\frac{1}{2}\textbf{e}^T_x \textbf{M}_{\sigma_x(B)} \textbf{e}_x)$.
\ENDFOR
\STATE \textbf{Output} $\Phi(S_y)\leftarrow \frac{1}{m}\sum\limits_{1\leq x\leq m}\phi_x(S_y)$.
\end{algorithmic}
\end{algorithm}
\subsection{Our Computing Approach to $\Phi(S_i)$ Accommodates Various Generalizations}
\begin{remark}[On our computing approach]\label{remark_on_computings}
    Our computing approach accommodates various generalized scenarios, e.g., 
    \begin{itemize}
    \item \textit{Sensing-graph adaptive.}  If the weight of PoI information varies across different edges, we modify the corresponding entry in our sensing graph matrix $\textbf{A}$ to reflect the specific weight. If the sensing graph $G_1$ contains loops (i.e., self-connections), we update the corresponding diagonal entries in our sensing matrix $\textbf{A}$ to one.
         
        \item \textit{Social-information-dissemination adaptive.} To accommodate the desired number of hops for information dissemination in the social graph, we have the flexibility to adjust the social connections of specific users by modifying the corresponding entries in their social matrix $B$.
    \end{itemize}
\end{remark}

\subsection{Discussions on Our Theoretical Guarantee in Theorem~\ref{ub_OBJ1F}.}
\begin{figure}[h]
     \centering    
     \includegraphics[width=5.5cm]{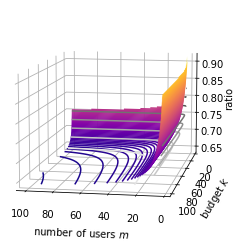}
     \caption{Visualization of the approximation guarantee for Algorithm~\ref{greedyalgforfiniteidea1} \textit{versus} budget $k$ and the number $m$ of selected users.}
\label{ratio_plot_greedyhire}
 \end{figure}
Fig.~\ref{ratio_plot_greedyhire} numerically depicts an overview of our theoretical guarantee versus the budget $k$ and the number $m$ of selected users, from which we can tell that our Algorithm~\ref{greedyalgforfiniteidea1} particularly fits small-size networks or low budgets. Even in a dense network or with a large budget, Fig.~\ref{ratio_plot_greedyhire} shows that  Algorithm~\ref{greedyalgforfiniteidea1} still guarantees at least a ($1-\frac{1}{e}$) proportion of the optimal welfare. 


\section{Omitted Proofs}\label{appendix_omittedproofs}
\subsection{Omitted Proof in Proposition~\ref{nphardnessingeneral}}
In the following Theorem~\ref{nphardness01}, we first establish the NP-hardness under the static crowd-sensing scenario of the problem, which sheds light on the NP-hardness proof for the mobile crowd-sensing scenario.
\begin{theorem}\label{nphardness01}
 The static crowd-sensing scenario of the social-enhanced PoI sharing problem is NP-hard.
\end{theorem}
\begin{proof}[Proof of Theorem~\ref{nphardness01}]
We show the NP-hardness of the static crowd-sensing scenario by a reduction from the decision version of the well-known NP-hard \textit{vertex cover problem} (VCP)\cite{hochbaum1982approximation}. Our following proof focuses on static crowd-sensing without PoI preference. That is every user $v_i\in V$ is interested in the PoI from the whole graph $G_1$, i.e., $E_1^{(i)}=E_1$.

To show our reduction, we construct from VCP's instance graph $G$ an example of our problem by setting 
$G_1=G$ and $G_2=(V, \varnothing)$. Accordingly, our objective function (\ref{obj_function}) can be written as
$$\Phi(S_k)=\frac{1}{m}\sum\limits_{v_i\in V}|E_1(\{v_i\}\cup S_k)|.$$

Next, we will show by two cases that 
VCP answers ``yes" if and only if our problem achieves  $\Phi(S_k)=|E_1|$. 

\noindent\textbf{Case 1. }
``$\Rightarrow$". If the VCP answers ``yes", there exists a subset $S'_k\subset V$ of size $k$ that touches each edge in $E_1$, i.e., $E_1(S'_k)=|E_1|$. Then, we have
\begin{equation}
\begin{split}
\Phi(S'_k)&=\frac{1}{m}\sum\limits_{v_i\in V}|E_1(\{v_i\}\cup S_k')|\\
&\geq \frac{1}{m}\sum\limits_{v_i\in V}|E_1(S'_k)|\\
&=|E_1|.
\end{split}
\end{equation}
Since $|E_1|$ is known to be an upper bound on the objective function (\ref{obj_function}), we have $\Phi(S_k')\leq |E_1|$. Thereby, we know  $\Phi(S_k')=|E_1|$ telling that
$S'_k$ is our optimal solution as well.

\noindent \textbf{Case 2.}
``$\Leftarrow$". If our problem finds $S_k\subseteq V$ that achieves
$\frac{1}{m}\sum\limits_{v_i\in V}|E_1(\{v_i\}\cup S_k)|=|E_1|$, the following holds for each $v_i\in V$, 
\begin{equation}\label{theorem1_eq02}
   |E_1(\{v_i\}\cup S_k)|=|E_1|, 
\end{equation}
which is also because $|E_1|$ upper bounds our  objective (\ref{obj_function}).
For a particular user $s_j\in S_k$, the above (\ref{theorem1_eq02}) implies
\begin{equation}
    E_1(\{s_j\}\cup S_k)= E_1(S_k)=|E_1|.
\end{equation}
In other words, $S_k$ touches every edge in $E_1$, and VCP answers ``yes" accordingly.

Cases 1 and 2 conclude this proof.
\end{proof}
Built upon Theorem~\ref{nphardness01}, we further show in the following theorem that our mobile crowd-sensing problem is also NP-hard, which also benefits from a special structure of the sensing graph we observe in a feasible reduction. 
\begin{theorem}\label{theorem_nphardness_hop_forward}
Our social-enhanced PoI sharing problem in the mobile crowd-sensing scenario is NP-hard.  
\end{theorem}
  \begin{proof}[Proof of Theorem~\ref{theorem_nphardness_hop_forward}]
 In the mobile crowd-sensing scenario, we show the NP-hardness of the problem by a reduction from the static crowd-sensing scenario without PoI preference. 
  Given an instance $I_{\rm o}(G_1(V,E_1),G_2(V,E_2)$ of the static crowd-sensing scenario without PoI preference, we construct an example $I_{\rm e}=\{G'_1,G_2\}$ for mobile crowd-sensing in the following manner: for each $v_i\in V$, we 
  \begin{itemize}
      \item construct another $|E_1|+n$ dummy nodes as summarized in the new set $V'_i=\{ v'_{i1},..., v'_{in},...., v'_{i(|E_1|+n)}\}$,
\item and construct edges as summarized in the following set $ E'_i=\{(v'_{ij},v'_{ij+1})|_{j=0,1,...,n-1},(v'_{in},v'_{ij})|_{j=n+1,...,n+|E_1|}\}$.
\end{itemize}
    So on and so forth, we obtain the sensing graph for the mobile crowd-sensing scenario as $ G'_1=(V\cup\bigcup\limits_{v_i\in V}V'_i,E_1\cup\bigcup\limits_{v_i\in V}E'_i)$. We illustrate in Fig.~\ref{pic_n_hop_nphardness} an example of the graph $G'_1$ that is constructed from a given graph $G_1$ in Fig.~\ref{systemfigure}. 
    \begin{figure}[!t]
\centering\includegraphics[width=6cm]{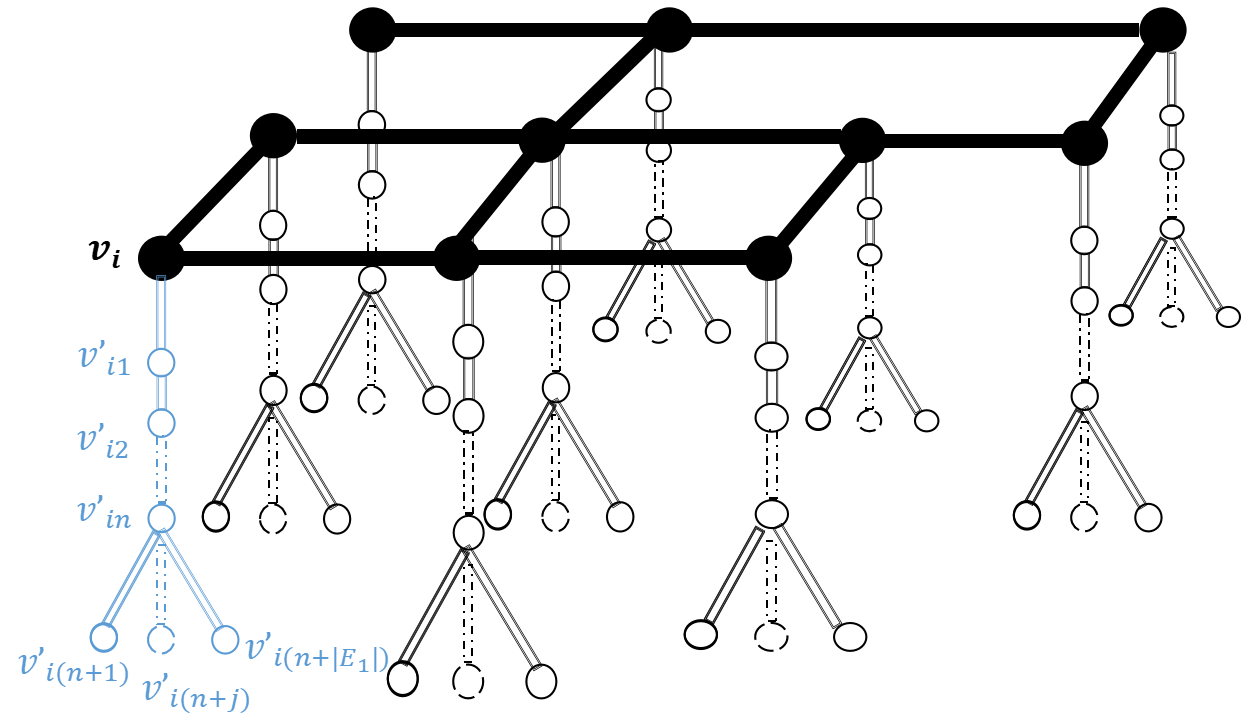}
    \caption{The graph constructed from the sensing graph in Fig. \ref{systemfigure}.}
    \label{pic_n_hop_nphardness}
\end{figure}
The following proposition provides a building block for NP-hardness proof, which reveals a connection between the optimal solutions to the static and mobile
scenarios. 
\begin{lemma}\label{prop01_in_nphardness_forwardsetting_new}
An optimal solution to the mobile crowd-sensing scenario with input $I_{\rm e}=(G'_1,G_2)$ $\iff$  an optimal solution to the static crowd-sensing scenario with input $I_{\rm o}=(G_1,G_2)$.
      \end{lemma}
    \begin{proof}[Proof of Lemma~\ref{prop01_in_nphardness_forwardsetting_new}]
         Note that only nodes in $V$ could there exist users stand-by for collecting data. For a selected user, say $v_i\in V$, we consider the following $n$-hop-forward move in the mobile crowd-sensing scenario:
\begin{equation}\label{best_n_hop_forward_movement}    (v_i,v'_{i1},v'_{i2},...,v'_{in}).
\end{equation}
After the above $n$-hop-forward move, $v_i$ could collect PoI information from both edges $E_1$ that are incident to $v_i$ and those edges in $E'_i$ (which only correspond to $v_i$ herself). We notice the following:
         \begin{equation}
             |E'_i|=n+|E_1|>|E_1|.
         \end{equation}
This implies that the best $n$-hop-forward move for $v_i$, in terms of increasing her contributed data to the system, should follow the path in  (\ref{best_n_hop_forward_movement}). Hence, regardless of which $k$ users are selected, (\ref{obj_function_pre_nhop_forwarded}) consistently indicates that
 \begin{equation}
\phi_i(S_k)=|E_1(\{v_i\}\cup N_{2}(v_i)\cup S_k|+k\cdot (n+|E_1|).
\end{equation}
 Further, by (\ref{obj_function}), we have
 \begin{equation}
    \Phi(S_k)=\underbrace{\frac{nk+k|E_1|}{m}}_{\rm constant}+\underbrace{\frac{1}{m}\sum\limits_{v_i\in V}|E_1(\{v_i\}\cup N_{2}(v_i)\cup S_k)|}_{{\rm objective}\;(\ref{obj_function_pre})}.
 \end{equation}
 This proves Proposition~\ref{prop01_in_nphardness_forwardsetting_new}.
      \end{proof}   
With Lemma \ref{prop01_in_nphardness_forwardsetting_new}, the NP-hardness for the mobile crowd-sensing scenario holds readily. 

This completes the proof.
\end{proof}

\subsection{Omitted Proof in Lemma \ref{lemma_01}}\label{appendix_lemma_01}
\begin{proof}
We note the following equations, in which the first equation is due to 
the matrix multiplication principle and the second equation holds by reorganizing the right-hand-side terms of the first equation.
\begin{eqnarray}    \label{lemma01_eq01}
\sum\limits_{1\leq j\leq m}c_{j,x}&=&\underbrace{(a_{1,1}b_{1,x}+a_{1,2}b_{2,x}+...+a_{1,m}b_{m,x})}_{i.e., c_{1,x}}\nonumber    \\
&\;&+\underbrace{(a_{2,1}b_{1,x}+a_{2,2}b_{2,x}+...+a_{2,m}b_{m,x})}_{i.e., c_{2,x}}  \nonumber    \\
&\;&\;\;\vdots \nonumber    \\
&\;&+\underbrace{(a_{m,1}b_{1,x}+a_{m,2}b_{2,x}+...+a_{m,m}b_{m,x})}_{i.e., c_{m,x}}\nonumber\\
&=&\underbrace{a_{1,1}b_{1,x}+a_{2,1}b_{1,x}+...+a_{m,1}b_{1,x}}_{i.e., b_{1,x}\cdot(\sum_{1\leq j\leq m}a_{j,t})}\nonumber    \\
&\;&+\underbrace{(a_{1,2}b_{2,x}+a_{2,2}b_{2,x}+...+a_{m,2}b_{2,x})}_{i.e., b_{2,x}\cdot(\sum_{1\leq j\leq m}a_{j,t})} \nonumber    \\
&\;&\;\;\vdots \nonumber    \\
&\;&+\underbrace{(a_{1,m}b_{m,x}+a_{2,m}b_{m,x}+...+a_{m,m}b_{m,x})}_{i.e., b_{m,x}\cdot(\sum_{1\leq j\leq m}a_{j,t})}\nonumber\\
&=&\sum\limits_{1\leq t\leq m}b_{t,x}\cdot(\sum_{1\leq j\leq m}a_{j,t}).
\end{eqnarray}
Since $\textbf{A}$ is a symmetric matrix, see Equation (\ref{metrixa_eq}), $\sum_{1\leq j\leq m}a_{j,t}$ indicates the number of edges that are incident to user $v_t\in V$.
Due to Equation (\ref{matrixb_eq}), we know that $b_{t,x}$ equals one only when user $v_t$ is either exactly $v_x$ or a neighbor of $v_x$. In other words, $\sum\limits_{1\leq t\leq m}b_{t,x}\cdot(\sum_{1\leq j\leq m}a_{j,t})$ counts the number of edges that are incident to nodes in $\{v_x\}\cup N_2{v_x}$ respectively and sums up all those counts together. The following two propositions support our result.
\begin{proposition} \label{prop_01}
$\sum\limits_{1\leq t\leq m}b_{t,x}\cdot(\sum\limits_{1\leq j\leq m}a_{j,t})$ counts an edge that connects two nodes of $\{v_x\}\cup N_2(v_x)$ two times, counts an edge that connects a node of $\{v_x\}\cup N_2(v_x)$ and a node of $V-\{v_x\}\cup N_2(v_x)$ one time, and counts an edge that connects two nodes of $V-\{v_x\}\cup N_2(v_x)$ zero time.
\end{proposition}
\begin{proof}
Note that each edge $(v_i,v_j)$ could be counted at most twice by counting edges that are incident to the two end nodes $v_i$ and $v_j$, respectively. 
Due to our model, an edge that connects two nodes in $V-\{v_x\}\cup N_2(v_x)$ is available to $v_x$, and hence, could not be counted in  $\sum\limits_{1\leq j\leq m}c_{j,x}$. An edge that connects a node in $\{v_x\}\cup N_2(v_x)$ and a node in $V-\{v_x\}\cup N_2(v_x)$ is counted only once when counting the edges that are incident to the node in $\{v_x\}\cup N_2(v_x)$. An edge that connects to two nodes in $\{v_x\}\cup N_2(v_x)$ will be counted twice when counting edges that are incident to the two nodes respectively.
\end{proof}

\begin{proposition}\label{prop_02}
     $\frac{1}{2}\textbf{e}_x^T \textbf{M}_{\sigma_x(B)} \textbf{e}_x$ indicates the number of different edges counted in $\sum\limits_{1\leq t\leq m}b_{t,x}\cdot(\sum\limits_{1\leq j\leq m}a_{j,t})$ that connect two nodes in $\{v_x\}\cup N_2(v_x)$.
\end{proposition}
\begin{proof}
 For each user, say $v_x\in V$, her corresponding matrix $\textbf{M}_{\sigma_x(B)}$ consists of those rows and columns of the matrix $\textbf{A}$ whose indices in matrix $\textbf{C}$ equal 1. In other words, $\textbf{M}_{\sigma_x(B)}$ reflects the road connection in $G_1$ among (users) nodes of $\{v_x\}\cup N_2(v_x)$ (which includes $v_x$ and her neighbors in 
 graph $G_2$).  Thus, an entry $(i,j)$ of $\textbf{M}_{\sigma_x(B)}$ equals one if and only if $v_i$ and $v_j$ are in the intersection $\{v_x\}\cup N_2(v_x)$ and are connected in graph $G_1$. Since $\textbf{A}$ is a symmetric matrix, its minor matrix $\textbf{M}_{\sigma_x(B)}$ is also symmetric.
Hence, $\textbf{e}_x^T \textbf{M}_{\sigma_x(B)} \textbf{e}_x$, which is the sum of all entries of $\textbf{M}_{\sigma_x(B)}$, doubles the number of different edges that connect two nodes in $\{v_x\}\cup N_2(v_x)$. This completes the proof.
\end{proof}
Therefore, Lemma~\ref{lemma_01} follows accordingly.
\end{proof}
\subsection{Omitted Proof in Proposition~\ref{theorem_02}}\label{appendix_theorem_02}
\begin{proof}
    Due to objective function (\ref{obj_function}) and Lemma \ref{lemma1_formulation}, we have 
    \begin{equation}
    \begin{split}
        \Phi(\varnothing)&=\frac{1}{m}\sum\limits_{v_i\in V}   |E_1(\{v_i\}\cup N_1(v_i))|\\
        &=\frac{1}{m}\sum\limits_{1\leq i\leq m}  \sum\limits_{1\leq j\leq m}c_{j,i}-\frac{1}{m}\sum\limits_{v_i\in V} \frac{1}{2}\textbf{e}_i^T \textbf{M}_{\sigma_i(B)} \textbf{e}_i\\
        &=\frac{1}{m}\textbf{e}^T\textbf{C}\textbf{e}-\frac{1}{2m}\sum\limits_{v_x\in V}\textbf{e}^T_x \textbf{M}_{\sigma_x(B)} \textbf{e}_x.
        \end{split}
    \end{equation}
    This completes the proof.
\end{proof}
\subsection{Omitted Proof in Lemma~\ref{lemma_02}}\label{appendix_lemma_02}
\begin{proof}
For any two node subsets $V',V''\subseteq V$ and any subset $E_1^{(i)}\subseteq E_1$,  we observe \begin{equation}\label{set_to_value_eq}
        |E_i(V'\cup V'')|=|E_i(V')|+|E_i(V'')|-|E_i(V'\cap V'')|,
    \end{equation}
and \begin{equation}\label{set_to_value_eq_fact2}
    \begin{split}
        &|E_1^{(i)}\cap E_i(V'\cup V'')|+|E_1^{(i)}\cap E_i(V'\cap V'')|\\&=|E_1^{(i)}\cap E_i(V')|+|E_1^{(i)}\cap E_i(V'')|.
         \end{split}
    \end{equation}
In the following two cases, we prove that function $\Phi(S_i)$ in (\ref{obj_function}) achieves monotonicity and submodularity, respectively.

\noindent \textbf{Result 1.} (\textit{monotonicity} of $\phi_i(S)$)

For an arbitrary subset $S\subseteq V$ and an arbitrary element $v\in V-S$, the following holds readily due to Equation (\ref{obj_function_pre}):

\begin{equation}\label{lemma1_eq03}
\begin{split}
    &\phi_i(S\cup\{v\})\\
    &=|E_1(\{v_i\}\cup N_{2}(v_i)\cup S\cup\{v\})|  \\
    &=|E_1(\{v_i\}\cup N_{2}(v_i)\cup S)\cup E_1(\{v\})|\\
    &\geq |E_1(\{v_i\}\cup N_{2}(v_i)\cup S)| =  \phi_i(S),
\end{split}
\end{equation}

\noindent
for any user $v_i\in V$, telling that $\phi_i(S)$ is non-decreasing over subsets  $S\subseteq V$. Further by Equation (\ref{obj_function}), we know $\Phi_i(S)$ is non-decreasing in $S\subseteq V$. 

\noindent \textbf{Result 2.} (\textit{submodularity} of $\phi_i(\cdot)$)

For any two subsets $S_1,S_2\subseteq V$ with $S_1\subseteq S_2$, and an element $v\in V$, we have

\begin{equation*}\label{lemma1_eq02}
\begin{split}
    &\phi_i(S_2\cup \{v\})-\phi(S_2)\\
    &=|E_1(\{v_i\}\cup N_{2}(v_i)\cup S_2\cup \{v\})|\\
    &\quad\;- |E_1(\{v_i\}\cup N_{2}(v_i)\cup S_2)|\\
    &=|E_1(\{v\})|-|E_1(\{v_i\}\cup N_{2}(v_i)\cup S_2)\cap \{v\})|\\
    &\leq |E_1(\{v\})|-|E_1(\{v_i\}\cup N_{2}(v_i)\cup S_1)\cap \{v\})|\\
    &=\phi_i(S_1\cup \{v\})-\phi_i(S_1),
\end{split}
\end{equation*}

\noindent
where the first equation is due to Equation (\ref{obj_function_pre}),     
 the first inequality holds since $S_1\subseteq S_2$, and the second and the last equations hold by (\ref{set_to_value_eq}). Therefore, each $\phi_i(\cdot)$ is submodular over subsets of $V$. Further by Equation (\ref{obj_function}), we know $\Phi(\cdot)$ is submodular as well.

Due to (\ref{set_to_value_eq_fact2}) and the fact that $E_1^{(i)}$ is a fixed edge subset, each $\phi_i(\cdot)$ meets the monotonicity and submodularity in the static crowd-sensing scenario. Further, the summation over every $\phi_i(\cdot)$, which is $\Phi(\cdot)$ is submodular and non-decreasing as well.

This proves Lemma~\ref{lemma_02}
\end{proof}
\subsection{Omitted Proof in Theorem~\ref{ub_OBJ1F}}
\begin{proof}
To start with, we discuss the \textit{time complexities} of our algorithms.  It is clear that the time complexity of Algorithm \ref{greedyalgforfiniteidea1} largely depends on the single step in its ``for-loop", which is further determined by Algorithms \ref{compputing_phi_varnothing} and \ref{compputing_phi_si}. Before discussing the specific time complexities of our algorithms in different scenarios, we note that a very recent work \cite{alman2021refined} provides an $O(m^{2.372})$-time matrix multiplication approach for computing two $m\times m$ matrices, which accelerate our approaches as well.

 (\textit{Running time for static crowd-sensing without PoI preference}). The running time of Algorithm \ref{compputing_phi_varnothing} is dominated by the step of matrix multiplication between two $m\times m$ matrices $A$ and $B$, which can be done in $O(m^{2.372})$-time. In Algorithm \ref{compputing_phi_si}, the initialization step in obtaining $A$ and $B$ can be done in $O(m^2)$ time by checking each entry of the matrix, the ``for-loop" can be done within $O(km)$-time, and $C$ can be got in $O(m^{2.372})$-time, leading to an overall running time of $O(m^{2.372})$. Further, to compute each $s_i$, Algorithm~\ref{greedyalgforfiniteidea1} needs to check ($m-i+1$) possibility in the single step of the ``for-loop" of Algorithm \ref{greedyalgforfiniteidea1}. In total, Algorithm~\ref{greedyalgforfiniteidea1} runs in $O(k(m-k)m^{2.372})$-time for the static crowd-sensing scenario without PoI preference.

(\textit{Running time for static crowd-sensing with PoI preference}). Due to Algorithm \ref{compputing_phi_varnothing_restricted}, the part of running time in computing each $\Phi(S_i)$ is dominated by the second ``for-loop", in which each iteration takes $O(m^{2.372})$-time leading to an overall running time of $O(m^{3.372})$. Further, to compute each $s_i$, Algorithm~\ref{greedyalgforfiniteidea1} needs to check ($m-i+1$) possibilities in which each takes time $O(m^{3.372})$. In total, Algorithm~\ref{greedyalgforfiniteidea1} runs in $O((m-k)km^{3.372})$-time. 


Benefits from a recent matrix multiplication approach \cite{alman2021refined}, our matrix computation-based algorithm can complete each step in its ``for-loop" in $O(m^{2.372})$-time. In total, our Algorithm~\ref{greedyalgforfiniteidea1}  runs in $O(k(m-k)m^{2.2372})$-time. 

To show the approximation guaranteed by Algorithm~\ref{greedyalgforfiniteidea1}, we introduce a new function $\widetilde{\Phi}(S_k)$ as follows:
\begin{equation}\label{converted_obj}
    \widetilde{\Phi}(S_k)= \Phi(S_k)-\Phi(\varnothing),
\end{equation}
where $\Phi(\varnothing)$ and  $\Phi(S_k)$ are returned by subroutines Algorithms~\ref{compputing_phi_varnothing} and \ref{compputing_phi_si}, respectively.

By Lemma \ref{lemma_02}, one can easily check that $\widetilde{\Phi}(S_k)$ is non-decreasing, sub-modular and $ \widetilde{\Phi}(\varnothing)=0$. Moreover, it follows
\begin{equation}
\begin{split}
    &\arg\max\limits_{v_x\in V}\left\{\widetilde{\Phi}(S_{i-1}\cup\{s_x\})-\widetilde{\Phi}(S_{i-1})\right\}\\&=\arg\max\limits_{v_x\in V}\left\{{\Phi}(S_{i-1}\cup\{s_x\})-{\Phi}(S_{i-1})\right\}.
    \end{split}
\end{equation}
In the ``for-loop" part of Algorithm~\ref{greedyalgforfiniteidea1}, by replacing function $\Phi(\cdot)$ with $\widetilde{\Phi}(\cdot)$, the following holds by Lemma~\ref{submodular_supportlemma}:
\begin{equation}\label{theorem3_eq03}
    \widetilde{\Phi}(S_k)\geq [1-(\frac{k-1}{k})^k]\cdot \widetilde{\Phi}(S_k^*),
\end{equation}
which implies  
\begin{equation}\label{theorem3_eq04}
    \Phi(S_k)\geq [1-(\frac{k-1}{k})^k]\cdot \Phi(S_k^*)+(\frac{k-1}{k})^k\Phi(\varnothing),
\end{equation}
due to the fact that $\Phi(\varnothing)$ remains unchanged regardless of the specific sets $S_k$ and $S_k^*$.

Hence, the approximation ratio follows
\begin{equation}\label{theorem3_eq05}
\begin{split}
    \rho \geq  [1-(\frac{k-1}{k})^k]+(\frac{k-1}{k})^k\frac{\Phi(\varnothing)}{\Phi(S_k^*)}.
\end{split}
\end{equation}
To reveal our approximation guarantee, we need further bound $\frac{\Phi(\varnothing)}{\Phi(S_k^*)}$ in (\ref{theorem3_eq05}). Thereby, we further discuss two cases.

\noindent\textbf{Case 1.} In the static crowd-sensing scenario without PoI preference, we notice, on the one hand, that $|E_1(V)|$ is a natural upper bound on $|\Phi(S_k^*)|$, i.e., 
\begin{equation}\label{theorem3_eq_bound_opt}
    |\Phi(S_k^*)|\leq |E_1(V)|.
\end{equation}
On the other hand,
\begin{equation}\label{theorem3_eq06}
\begin{split}
\Phi(\varnothing)&=\frac{1}{m}\sum\limits_{v_i\in V}|E_1(\{v_i\}\cup N_{2}(v_i))|\\
&\geq \frac{1}{m}\sum\limits_{v_i\in V}|E_1(\{v_i\})|=\frac{2|E_1|}{m},
\end{split}
\end{equation}
where the first equation is due to Equations (\ref{obj_function_pre}) and (\ref{obj_function}),  the first inequality holds by the monotonicity of the function $\Phi(\cdot)$, and the second equation holds since each edge $(v_x,v_y)\in V$ is counted twice in the summation $|\sum\limits_{v_i\in V}|E_1({v_i})|$ because it is considered separately for $E_1({v_x})$ and $E_1({v_y})$.
Plugging Inequalities (\ref{theorem3_eq_bound_opt}) and (\ref{theorem3_eq06}) in Inequality (\ref{theorem3_eq05}), yields
\begin{equation*}
\begin{split}
    \rho&\geq 1-(\frac{k-1}{k})^k+(\frac{k-1}{k})^k\frac{\frac{2|E_1|}{m}}{|E_1|}=1-\frac{m-2}{m}(\frac{k-1}{k})^k.
\end{split}
\end{equation*}

\noindent\textbf{Case 2.} For static crowd-sensing with 
 PoI preference,
we have
\begin{equation}\label{theorem_case2_basic eq}
\begin{split}
\frac{\Phi(\varnothing)}{\Phi(S_k^*)}&\geq\frac{\frac{1}{m}\sum\limits_{v_i\in V}|E_1^{(i)}\cap E_1(\{v_i\}\cup N_{2}(v_i))|}{\frac{1}{m}\sum\limits_{v_i\in V}|E_1^{(i)}|}\\
&\geq \frac{\frac{1}{m}\sum\limits_{v_i\in V}| E_1(\{v_i\})|}{\frac{1}{m}\sum\limits_{v_i\in V}|E_1^{(i)}|}\geq \frac{2}{m},
\end{split}
\end{equation}
where the first inequality holds because ${\frac{1}{m}\sum\limits_{v_i\in V}|E_1^{(i)}|}$ is a natural upper bound on the optimal solution in this setting, and 
 {\color{black} the second inequality holds by $E_1(v_i)\subseteq E_1^{(i)}$, and the last one inequality holds by  $E_1^{(i)} \subseteq E_1$.}

Finally, we have by plugging (\ref{theorem_case2_basic eq}) in (\ref{theorem3_eq05}) that
\begin{equation}
    \rho\geq 1-\frac{m-2}{m}(\frac{k-1}{k})^k.
\end{equation} 
Cases 1-2 conclude the proof.
\end{proof}

\subsection{Omitted Proof in Theorem~\ref{theorem_n_hop_forward}}
\begin{proof}
{\color{black}We start by considering the time complexity of Algorithm~\ref{greedyalgforfiniteidea1_augmented}.
Due to our matrix multiplication approach in Theorem~\ref{thm4.3_for_computing_Phi}, we can obtain each $\Phi(S_i)$ in $O(\varpi^{2.372})$-time  where $\varpi=|V\cup V'|$. To compute each $p_i$, Line 3 in our Algorithm~\ref{greedyalgforfiniteidea1_augmented} takes at most ($\tau-i+1$) iterations in which each iteration takes $O(\varpi^{2.372})$-time. Thereby, the whole ``for-loop" of Algorithm~\ref{greedyalgforfiniteidea1_augmented} runs in  $O(k(\tau-k)\varpi^{2.372})$-time. Together with its initial step for constructing $\mathcal{P}$ in 
$O(m(|V'|+m)^2)$-time, Algorithm~\ref{greedyalgforfiniteidea1_augmented} totally runs in $O(k(\varpi^n-k)\varpi^{2.372})$-time.}

Similar to the proof of Lemma \ref{lemma_02}, it is easy to check that our new objective (\ref{new_obj_nhop}) for the mobile crowd-sensing scenario remains monotonicity and submodularity.  Given an instance $I=(G_1,G_2)$ of our problem, denote $I_x$ as the augmented instance with factor $g=x$. The only difference between $I_x$ and $I$ is that each user node of the sensing graph $G_1$ in instance $I_x$  contains exactly $x$ identical users instead of exactly one in the original instance $I$. 
Denote $P$ and $P^*$ as the set of paths that are generated by our Algorithm~\ref{greedyalgforfiniteidea1_augmented} with augmentation $g=x$ and an optimal solution, respectively.

Before showing our approximation result,  we note a feature of our Algorithm~\ref{greedyalgforfiniteidea1_augmented}, which possibly narrows the search space by the ``if-else" of Algorithm \ref{greedyalgforfiniteidea1_augmented} in some of its selection iterations.  This feature makes the submodularity result in  Lemma~\ref{submodular_supportlemma} not be applied straightforwardly in deriving the approximation guarantee of Algorithm~\ref{greedyalgforfiniteidea1_augmented}. In other words, the result in Lemma~\ref{submodular_supportlemma}
could be applied only when the search space does not change (i.e., the ``if-else" branches of Algorithm \ref{greedyalgforfiniteidea1_augmented} is reduced). As such, our Algorithm~\ref{greedyalgforfiniteidea1_augmented}  with augmentation factor $g=k$ does not change its search space. In the following, we discuss two cases.

\noindent\textbf{Case 1}. ($g=k$).
 The approximation ratio of our Algorithm~\ref{greedyalgforfiniteidea1_augmented}  with augmentation $g=k$ follows
\begin{equation}\label{equations_augmentation00}
\begin{split}
\frac{\Phi(P(I_k))}{\Phi(P^*(I))}
&\geq\frac{\Phi(P(I_k))}{\Phi(P^*(I_k))} \\
&\geq [1-(\frac{k-1}{k})^k]+(\frac{k-1}{k})^k\frac{\Phi(\varnothing)}{\Phi(P^*(I_k))}\\
&\geq 1-\frac{\varpi-2}{\varpi}(\frac{k-1}{k})^k,
\end{split}
\end{equation}
in which the first inequality holds since resource augmentation does not influence an optimal solution, the second inequality holds according to (\ref{theorem3_eq05}), and the last inequality holds by (\ref{theorem3_eq06}) and $\Phi(P^*(I_k))\leq \varpi$.

\noindent\textbf{Case 2.} ($1\leq g\leq k-1$).

To show the approximation guarantee of this case, let us consider such a solution from the following steps: 
\begin{itemize}
    \item \textit{First}, apply Algorithm~\ref{greedyalgforfiniteidea1_augmented} with augmentation $k$ to obtain  $P(I_k)$ as a subroutine. 
    \item {\color{black}\textit{Then}, categorize paths in $P(I_k)$ into distinct classes based on their start nodes, with each class forming a separate set $C_i$ in $\mathcal{C}$. As a consequence, any two paths from the same class $C_i\in \mathcal{C}$ possess the same start node, while any two paths from two different classes $C_i\in \mathcal{C}$ and $C_j\in \mathcal{C}$ hold different start nodes.}
    \item \textit{Next}, from each distinct class $C_i\in \mathcal{C}$, we output the first path selected to by Algorithm~\ref{greedyalgforfiniteidea1_augmented}, forming set $\overline{P}_{C_i}$. By summarizing all output paths from sets $\overline{P}_{C_i}$, we obtain a new solution
    $\overline{P}(I_k)$, i.e., $\overline{P}(I_k)\triangleq\bigcup\limits_{C_i\in \mathcal{C}}\overline{P}(C_i)$.
\end{itemize}
We refer to the above solution $\overline{P}(I_k)$ as an \textit{intermediate solution} to Algorithm~\ref{greedyalgforfiniteidea1_augmented}. In our following proof, we use this intermediate solution $\overline{P}(I_k)$ as a bridge that connects the welfare $P(I_g)$ of our Algorithm~\ref{greedyalgforfiniteidea1_augmented} with any augmentation $g$ and the welfare $P(I_k)$ of  Algorithm~\ref{greedyalgforfiniteidea1_augmented} with a specific augmentation factor $k$. 
Note that $\overline{P}(I_k)$ is actually a solution to an augmented instance $I_g$ with factor $g$, which implies the following due to the selection process of Algorithm~\ref{greedyalgforfiniteidea1_augmented}:
\begin{equation}\label{appendix_theorem51_eq_01}
    \Phi(P(I_g))\geq \Phi(\overline{P}(I_k)).
\end{equation}
Notice that the earlier a path is selected to a path class  $\mathcal{C}_i$, the more welfare benefits the path contributes to the system. This tells the following relation.
\begin{equation}\label{appendix_theorem51_eq_02}
     \Phi(\overline{P}(C_i))\geq \frac{g}{k}  \Phi(C_i),
\end{equation}
which implies
\begin{equation}\label{appendix_theorem51_eq_03}
    \Phi(\overline{P}(I_k))\geq \frac{g}{k}  \Phi(P(I_k)).
\end{equation}
By combining Inequalities (\ref{appendix_theorem51_eq_01}) and (\ref{appendix_theorem51_eq_03}), we have  
\begin{equation}\label{equations_augmentation01}
    \Phi(P(I_g))\geq \frac{g\Phi(P(I_k))}{k}
\end{equation}
Finally, the approximation ratio of Algorithm~\ref{greedyalgforfiniteidea1_augmented} 
 follows
 \begin{equation}
\begin{split}
\frac{\Phi(P(I_g))}{\Phi(P^*(I))}&\geq\frac{g \Phi(P(I_k))}{k\Phi(P^*(I))}\geq\frac{g \Phi(P(I_k))}{k\Phi(P^*(I_k))}\\&\geq \frac{g}{k}[1-\frac{\varpi-2}{\varpi}(\frac{k-1}{k})^k],
\end{split}
\end{equation}
in which the first inequality holds by (\ref{equations_augmentation01}), the second inequality holds as $\Phi(P^*(I))\leq \Phi(P^*(I_k))$ and the third inequality holds by (\ref{equations_augmentation00}).
\end{proof}
\subsection{Omitted Proof in Theorem~\ref{improvedratio_nhop}}
\begin{proof}
The running time of our Algorithm~\ref{adjusted_greedyhireg_alg} is dominated by Step 3 for constructing a solution $P$ for the augmented instance $I_k$. Hence, Algorithm~\ref{adjusted_greedyhireg_alg} runs in time $O(k(\tau-k)\varpi^{2.372})$ due to Theorem \ref{improvedratio_nhop}). The following two propositions (as proved in Appendices \ref{prop_adjusted_nhop_01} and \ref{prop_adjusted_nhop_02}) further reveal our approximations.

\begin{proposition}\label{prop_adjusted_nhop_01}
 $P'$ is a solution to the original instance $I$ (that is, without augmentation), i.e., any two paths in $P'$ start from different user nodes.
\end{proposition}
\begin{proof}[Proof of Proposition \ref{prop_adjusted_nhop_01}]
According to the partition rule in Step 4 of Algorithm \ref{adjusted_greedyhireg_alg},
paths in different partitioned classes start from different user nodes, which implies that those paths included in the first ``for loop" of Algorithm \ref{adjusted_greedyhireg_alg} start from distinctive user nodes. Further, due to the ``if" condition of the second ``nested-for-loop", all the $k$ paths included in $P'$ start from $k$ distinctive user nodes. The proposition follows.
\end{proof}

\begin{proposition}\label{prop_adjusted_nhop_02}
 $\Phi(P(I_k))=\Phi(P')$.
 \end{proposition}
 \begin{proof}[Proof of Proposition \ref{prop_adjusted_nhop_02}]
    As per our transformed Problem (\ref{new_obj_nhop})-(\ref{new_obj_nhop_constraint3}), we know, as mentioned earlier, that different sets of selected paths contribute the same as long as they visit the same set of nodes in the sensing graph $G_1$. Denote $V(P')$ and $V(P)$ as the set of nodes that are visited by paths in sets $P'$ and $P$, respectively. To further show the proposition, it is sufficient to prove 
    \begin{equation}\label{eq0001_adjusted_theorem}
        V(P')\subseteq V(P).
    \end{equation}
    Note that in each path class $C_i\in \mathcal{C}$, the first path selected by  Algorithm \ref{adjusted_greedyhireg_alg} is also included in the solution $P'$ of  \textsc{Adjusted-GPS}. For each other path of $p$ in some class $C_i\in \mathcal{C}$, we discuss two cases as follows.
    
    \textbf{Case 1.} (There exist nodes in $p$ that do not appear in start nodes of paths included in $P'$ at the moment of selection.)

    Note that  \textsc{Adjusted-GPS} always selects the first one $p({\rm first})$ of such nodes, implying that all those nodes prior to the selected one in $p$ are visited by paths in the current $P'$. Further,  \textsc{Adjusted-GPS} includes to $P'$ a new path that is constructed by extending the sub-path of $p$ starting from node $p({\rm first})$ to $n$-size; this further guarantees that all nodes on $p$ are visited by some paths in $P'$. Then, the proposition follows.

    \textbf{Case 2.} Each node on $p$ appears as a start nodes of some path included in $P'$ at the moment of selection. The proposition follows accordingly.
\end{proof}

Proposition \ref{prop_adjusted_nhop_01} in conjunction with Proposition  \ref{prop_adjusted_nhop_02} tells 
\begin{equation}
\begin{split}
    \frac{\Phi(P'(I))}{\Phi(P^*(I))}\geq \frac{\Phi(P(I_k))}{\Phi(P^*(I_k))}\geq 1-\frac{\varpi-2}{\varpi}(\frac{k-1}{k})^k.
\end{split} 
\end{equation}
in which the first inequality is due to $\Phi(P^*(I))\leq \Phi(P^*(I_k))$ and Proposition \ref{prop_adjusted_nhop_02}, and the last inequality holds by (\ref{equations_augmentation00}).
\end{proof}
%





%
%
%
\end{document}